\documentclass[11pt,american,english]{article}
\usepackage{lmodern}
\usepackage[T1]{fontenc}
\usepackage[latin9]{inputenc}
\usepackage[letterpaper]{geometry}
\usepackage{verbatim}
\usepackage{amsmath}
\usepackage{amsthm}
\usepackage{amssymb}

\makeatletter
\numberwithin{equation}{section}
\numberwithin{figure}{section}
\theoremstyle{plain}
\newtheorem{thm}{\protect\theoremname}
\theoremstyle{plain}
\newtheorem{lem}[thm]{\protect\lemmaname}
\theoremstyle{plain}
\newtheorem{cor}[thm]{\protect\corollaryname}
\theoremstyle{remark}
\newtheorem*{rem*}{\protect\remarkname}
\theoremstyle{plain}
\newtheorem{fact}[thm]{\protect\factname}
\theoremstyle{definition}
\newtheorem{defn}[thm]{\protect\definitionname}

\usepackage{amsmath, amsthm, amssymb, amsfonts, bbm}
\usepackage[linktocpage=true,pdftex,bookmarks,colorlinks,allcolors=blue]{hyperref}
\usepackage{enumitem}
\usepackage{color}

\usepackage[lined,boxed,ruled,norelsize,algo2e,linesnumbered]{algorithm2e}

\newcommand{\sidford}[1]{{\linebreak \textcolor{green}{{\bf AS:} \em{#1}} \linebreak}}

\usepackage{tikz}
\makeatother

\usepackage{babel}
\addto\captionsamerican{\renewcommand{\corollaryname}{Corollary}}
\addto\captionsamerican{\renewcommand{\definitionname}{Definition}}
\addto\captionsamerican{\renewcommand{\factname}{Fact}}
\addto\captionsamerican{\renewcommand{\lemmaname}{Lemma}}
\addto\captionsamerican{\renewcommand{\remarkname}{Remark}}
\addto\captionsamerican{\renewcommand{\theoremname}{Theorem}}
\addto\captionsenglish{\renewcommand{\corollaryname}{Corollary}}
\addto\captionsenglish{\renewcommand{\definitionname}{Definition}}
\addto\captionsenglish{\renewcommand{\factname}{Fact}}
\addto\captionsenglish{\renewcommand{\lemmaname}{Lemma}}
\addto\captionsenglish{\renewcommand{\remarkname}{Remark}}
\addto\captionsenglish{\renewcommand{\theoremname}{Theorem}}
\providecommand{\corollaryname}{Corollary}
\providecommand{\definitionname}{Definition}
\providecommand{\factname}{Fact}
\providecommand{\lemmaname}{Lemma}
\providecommand{\remarkname}{Remark}
\providecommand{\theoremname}{Theorem}

\begin{document}
\selectlanguage{american}%
\global\long\def\M{\mathcal{M}}%
\global\long\def\P{\mathcal{P_{\M}}}%
\global\long\def\x{\textbf{\ensuremath{\mathbf{\mathbf{x}}}}}%
\global\long\def\z{\textbf{\ensuremath{\mathbf{z}}}}%
\global\long\def\I{\mathcal{\mathcal{I}}}%
\global\long\def\rk{\mathcal{\textsf{rank}}}%
\global\long\def\circ{\mathcal{\textsf{circuit}}}%
\global\long\def\free{\mathcal{\textsf{free}}}%
\global\long\def\otime{\mathcal{T}}%
\global\long\def\Tind{\mathcal{T}_{\mathcal{\textsf{ind}}}}%
\global\long\def\Trank{\mathcal{T}_{\rk}}%
\global\long\def\exchange{\textsf{exchange}}%
\global\long\def\I{\mathcal{\mathcal{I}}}%
\global\long\def\M{\mathcal{\mathcal{M}}}%
\global\long\def\Rn{\mathcal{\mathbb{R}}^{n}}%
\global\long\def\argmax{\mathrm{argmax}}%
\global\long\def\T{\mathcal{\mathcal{T}}}%
\global\long\def\Pa{\mathcal{P}_{\M_{1}}}%
\global\long\def\P{\mathcal{P}_{\M}}%
\global\long\def\Pb{\mathcal{P}_{\M_{2}}}%
\global\long\def\E{\mathbb{E}}%
\global\long\def\defeq{\stackrel{\mathrm{{\scriptscriptstyle def}}}{=}}%
\global\long\def\pac{\textsf{PackNum}}%
\global\long\def\R{\mathbb{R}}%
\global\long\def\GS{V}%
\global\long\def\eps{\epsilon}%
\global\long\def\ranktime{\Trank}%
\global\long\def\indeptime{\Tind}%
\global\long\def\M{\mathcal{M}}%
\global\long\def\P{\mathcal{P_{\M}}}%
\global\long\def\x{\textbf{\ensuremath{\mathbf{\mathbf{x}}}}}%
\global\long\def\z{\textbf{\ensuremath{\mathbf{z}}}}%
\global\long\def\I{\mathcal{\mathcal{I}}}%
\global\long\def\otime{\mathcal{T}}%
\global\long\def\Span{\mathcal{\textsf{span}}}%
\global\long\def\sidford#1{{\color{red}\textbf{sidford}: {#1}}}%

\title{\selectlanguage{english}%
Faster Matroid Intersection}
\author{Deeparnab Chakrabarty\thanks{Dartmouth College, deeparnab@dartmouth.edu}
\and Yin Tat Lee\thanks{University of Washington and Microsoft Research, yintat@uw.edu}
\and Aaron Sidford\thanks{Stanford University, sidford@stanford.edu. Research supported in part by NSF CAREER Award CCF-1844855.}
\and Sahil Singla\thanks{Princeton University and Institute for Advanced Study, singla@cs.princeton.edu}
\and Sam Chiu-wai Wong\thanks{Microsoft Research, samwon@microsoft.com}}
\maketitle
%\pagenumbering{roman}
\begin{abstract}
In this paper we consider the classic matroid intersection problem:
given two matroids $\M_{1}=(V,\I_{1})$ and $\M_{2}=(V,\I_{2})$ defined
over a common ground set $V$, compute a set $S\in\I_{1}\cap\I_{2}$
of largest possible cardinality, denoted by $r$. We consider this
problem both in the setting where each $\M_{i}$ is accessed through
an independence oracle, i.e. a routine which returns whether or not
a set $S\in\I_{i}$ in $\indeptime$ time, and the setting where each
$\M_{i}$ is accessed through a rank oracle, i.e. a routine which
returns the size of the largest independent subset of $S$ in $\M_{i}$
in $\ranktime$ time.\smallskip

In each setting we provide faster exact and approximate algorithms.
Given an independence oracle, we provide an \emph{exact} $O(nr\log r\cdot\indeptime)$
time algorithm. This improves upon previous best known running times
of $O(nr^{1.5}\cdot\indeptime)$ due to Cunningham in 1986$\ $\cite{Cunningham-SICOMP86}
and $\tilde{O}(n^{2}\cdot\indeptime+n^{3})$ due to Lee, Sidford,
and Wong in 2015~\cite{LSW15}. We also provide two algorithms which
compute a $(1-\epsilon$)-\emph{approximate} solution to matroid intersection
running in times $\tilde{O}(n^{1.5}/\eps^{1.5}\cdot\indeptime)$ and
$\tilde{O}((n^{2}r^{-1}\epsilon^{-2}+r^{1.5}\epsilon^{-4.5})\cdot\indeptime)$,
respectively. These results improve upon the $O(nr/\eps\cdot\indeptime)$-time
algorithm of Cunningham$\ $\cite{Cunningham-SICOMP86} as noted recently
by Chekuri and Quanrud~\cite{ChekuriQuanrud-SODA16}.\smallskip

Given a rank oracle, we provide algorithms with even better dependence
on $n$ and $r$. We provide an $O(n\sqrt{r}\log n\cdot\ranktime)$-time
exact algorithm and an $O(n\epsilon^{-1}\log n\cdot\ranktime)$-time
algorithm which obtains a $(1-\eps)$-approximation to the matroid
intersection problem. The former result improves over the $\tilde{O}(nr\cdot\ranktime+n^{3})$-time
algorithm by Lee, Sidford, and Wong~\cite{LSW15}. The rank oracle
is of particular interest as the matroid intersection problem with
this oracle is a special case (via Edmond's minimax characterization
of matroid intersection) of the submodular function minimization (SFM)
problem with an evaluation oracle, and understanding SFM query complexity
is an outstanding open question.
\end{abstract}
\clearpage

\setcounter{tocdepth}{3} \tableofcontents

\clearpage
%\pagenumbering{arabic} \setcounter{page}{1}

\pagebreak{}

\section{Introduction}

A matroid $\M=(\GS,\I)$ is an abstract set system defined over a
finite ground set (universe) $\GS$ of size $n$ where the collection
$\I\subseteq2^{V}$ of \emph{independent} sets satisfy two properties:
(a) $A\in\I$ implies every subset $B\subseteq A$ is also independent,
i.e., $B\in\I$, and (b) for any two sets $A,B\in\I$ with $|A|<|B|$,
there exists an element $e\in B\setminus A$ such that $A+e\in\I$.
Matroids are fundamental objects in combinatorics, and the abstract
definition above generalizes a wide range of concepts ranging from
acyclic graphs to linearly independent matrices.

Given two matroids $\M_{1}=(\GS,\I_{1})$ and $\M_{2}=(\GS,\I_{2})$
over the same ground set, the \emph{matroid intersection} problem
is to find a set $I\in\I_{1}\cap\I_{2}$ with the largest cardinality
$|I|$. This problem generalizes many important combinatorial optimization
problems such as bipartite matching, arborescences in digraphs, and
packing spanning trees. Unsurprisingly, this problem has applications
in areas such as electrical engineering~\cite{Mur09,Rec13} and network
coding~\cite{DFZ07,RSG10,DFZ11}. 

To define the problem algorithmically, one needs to specify access
to these matroids. In the literature it is common to assume access
to an \emph{independence oracle} which takes as input a subset $S\subseteq\GS$
and returns whether $S\in\I$ or not. We use $\indeptime$ to denote
the maximum time taken by such an oracle to answer a single query.
This raises an algorithmic question: in how few queries can the matroid
intersection problem be solved?

Starting from the work of Edmonds~\cite{Edmonds-Journal70}, many
polynomial time algorithms~\cite{AigD71,Law75,Cunningham-SICOMP86,GabXu89,ShiI95,LSW15}
have been proposed for the matroid intersection problem. The previous
state-of-the-art captured by two works. One is a classic $O(nr^{1.5}\cdot\indeptime)$
time combinatorial algorithm by Cunningham~\cite{Cunningham-SICOMP86}
where $r$ is the cardinality of the largest common independent set
of the two matroids. The second algorithm is an $\tilde{O}(n^{2}\cdot\indeptime+n^{3})$-time
algorithm by Lee, Sidford, and Wong~\cite{LSW15} based on fast algorithms
implementing the ellipsoid method. Note that this pays a higher overhead
than Cunningham's result and makes many more queries if $r\ll n^{2/3}$.
This raises the important question, \emph{``Is there an algorithm
which obtains the best of both results?''} Our first result answers
this affirmatively (restated as Theorem~\ref{thm:exact-indep} in
the main body).\footnote{In simultaneous and independent work \cite{nguyen2019arxiv}, Nguy\~{\^e}n also answered this question affirmatively obtaining an $O(nr (\log r)^2 \cdot\indeptime)$-time algorithm for solving the matroid intersection problem exactly.
See the remark at the end of Section~\ref{sec:Indep-Exact-Algorithm} for a comparison.}
\begin{thm}
\label{thm:exact} There is an $O(nr\log r\cdot\indeptime)$-time
algorithm to solve the matroid intersection problem exactly.
\end{thm}

Our next result looks at \emph{approximate matroid intersection}.
Due to both theoretical and practical reasons, there has been extensive
recent work~\cite{CKM+10,lee2013new,sherman2013nearly,kelner2014almost,peng2016approximate,ChekuriQuanrud-SODA16,HuaKK16,sherman2017area,sidford2018coordinate}
in trying to obtain faster $(1\pm\eps)$-approximation algorithms
for problems which already have polynomial time exact algorithms.
A $(1-\eps)$-approximate algorithm for matroid intersection would
return a set $I\in\I_{1}\cap\I_{2}$ with $|I|\geq(1-\eps)r$ where
$r$ is the cardinality of largest common independent set.

For matroid intersection, Cunningham's~\cite{Cunningham-SICOMP86}
algorithm already gives a $(1-\eps)$-approximate solution in $O(nr/\eps\cdot\indeptime)$
time. This observation was made explicit in a recent paper by Chekuri
and Quanrud~\cite{ChekuriQuanrud-SODA16}. Given our exact algorithm
result above, it is natural to wonder if one can obtain even faster approximation
algorithms. In particular, \emph{can there be subquadratic approximation
algorithms for the matroid intersection algorithm problem?} Our second
result gives an affirmative answer (restated as Theorem \ref{thm:approxMatrInters}
later).
\begin{thm}
\label{thm:apx} There is an $O(n^{1.5}\sqrt{\log r}/\eps^{1.5}\cdot\indeptime)$-time
algorithm to obtain a $(1-\eps)$-approximation to the matroid intersection
problem.
\end{thm}

Theorem~\ref{thm:apx} has no significant dependence on $r$, and indeed
when $r\ll\sqrt{n}$, as it stands, it is better to use the exact
algorithm. However, we can get a better result (restated as Theorem$\ $\ref{thm:apx-better-r}
in the main body) in the regime $\sqrt{n}\ll r\ll n$ ; concretely,
assume $r=\Theta(n^{c})$ where $\frac{1}{2}<c<1$. 

\begin{comment}
Further, we are able to combine our algorithm with first-order iterative
optimization methods to yield an even faster algorithm in certain
parameter regimes.
\end{comment}

\begin{thm}
\label{thm:sparse} There is an $\tilde{O}\left(\left(\frac{n^{2}}{r\epsilon^{2}}+\frac{r^{1.5}}{\epsilon^{4.5}}\right)\cdot\mathcal{T}_{ind}\right)$-time
algorithm to obtain a $(1-\eps)$-approximation to the matroid intersection
problem.
\end{thm}

\paragraph{The Rank Oracle.}

Another common model for accessing matroids is a \emph{rank oracle}.
Given a subset $S\subseteq\GS$, the rank oracle outputs $\rk(S)$,
i.e., the size of the maximum cardinality independent subset of $S$.
We let $\ranktime$ denote the maximum time taken by the rank oracle
to answer any query. The rank oracle is clearly at least as powerful
as the independence oracle.

Similar to the independence oracle, Lee, Sidford, and Wong~\cite{LSW15}
also gave an algorithm for the rank oracle with a runtime of $\tilde{O}(nr\cdot\ranktime+n^{3})$.
This suggests that perhaps matroid intersection can be solved strictly
faster in the rank oracle model.

One reason to look at the rank oracle is Edmonds' minimax theorem~\cite{Edmonds-Journal70}
which states that $\max_{I\in\I_{1}\cap\I_{2}}|I|=\min_{S\subseteq\GS}\left(\rk_{1}(S)+\rk_{2}(\GS\setminus S)\right)$.
Since the rank function is submodular, (the dual of) matroid intersection
is a special case of submodular function minimization (SFM) where
the function evaluation oracle corresponds to (two calls of) the rank
oracle. SFM is an extensively studied problem whose query complexity
is still an open question. In this light, understanding it for the
special case of matroid intersection becomes an important problem.

Another reason is that our main ideas for designing independence oracle-based algorithms arose from an understanding the question with rank oracle, and so the rank oracle
formulation constitutes a natural warm-up. In this paper we provide the following results
on solving matroid intersection (both approximately and exactly with
a rank oracle), improving over the $\tilde{O}(nr\cdot\ranktime+n^{3})$-time
algorithm by Lee, Sidford, and Wong~\cite{LSW15} (restated as Theorem~\ref{thm:exact_rank} and Theorem~\ref{thm:approx-rank} later).
\begin{thm}
\label{thm:rank} There is an $O(n\sqrt{r}\log n\cdot\ranktime)$-time
algorithm to exactly solve the matroid intersection problem. There
is an $O(n\epsilon^{-1}\log n\cdot\ranktime)$-time algorithm to obtain
a $(1-\eps)$-approximation to the matroid intersection problem.
\end{thm}

\subsection{Our Techniques}

At a high level our algorithms build upon existing combinatorial algorithms,
i.e. Edmonds~\cite{Edmonds-Journal70} and Cunningham~\cite{Cunningham-SICOMP86}.
We extend these algorithms leveraging the following key ideas: the
\emph{binary search} idea which allows fast exploration through the
exchange graph, the \emph{augmenting sets} methodology which allows
multiple parallel augmentations in the exchange graph, and \emph{first-order
methods} which allow efficient sparsification of the ground set. In
this section we explain each idea in greater depth, leaving the full
details to subsequent sections. We begin with a refresher and a high-level
overview of existing combinatorial algorithms.

\paragraph{The exchange graph and the algorithms of Edmonds and Cunningham.}

The key object behind almost all combinatorial algorithms for matroid
intersection is the \emph{exchange graph}. Given a current solution
$S\in\I_{1}\cap\I_{2}$, the exchange graph $G(S)$ is a \textit{directed}
bipartite graph where the endpoints of arcs correspond to valid exchanges
in respective matroids depending on the direction. There is a source
$s$ and sink $t$. The source is connected to all vertices in $V\setminus S$
which can be freely added in one matroid, and the sink is connected
from all vertices in $V\setminus S$ which can be freely added to
the other matroid. Much as in network flow theory, an improvement
in the size of $S$ arises on finding \emph{shortest} augmenting source-sink
paths in this exchange graph. At a very high level, there are $O(nr)$
possible edges in the graph (since $|S|\leq r$), each edge can be
constructed with $O(1)$ independence oracle queries, and in at most
$r$-augmentations one obtains the maximum sized common independent
set. This gives an $O(nr^{2}\cdot\indeptime)$-time algorithm~\cite{Edmonds-Journal70,AigD71,Law75}.

Cunningham's~\cite{Cunningham-SICOMP86} main idea was to implement
the above algorithm in \emph{phases}. This idea is similar to the
Hopcroft-Karp~\cite{HK-SICOMP73} idea for bipartite matching but
has many differences. Akin to~\cite{HK-SICOMP73}, in each ``early''
phase the algorithm tries to find many disjoint short augmenting paths;
however, not all of these can be augmented in parallel. Indeed, one
of the major divergences between bipartite matching and matroid intersection
is that even a single augmentation can completely change the exchange
graph. Although the augmentations cannot be done in parallel, Cunningham~\cite{Cunningham-SICOMP86}
shows how to spend only $O(nr)$ independence-oracle calls (each edge
is queried only once in a phase) to sequentially run all the augmentations
in a single phase; as in Hopcroft-Karp~\cite{HK-SICOMP73}, the early
phases lead to big improvements and the total number of phases is
at most $O(\sqrt{r})$. This leads to a total of $O(nr^{1.5}\cdot\indeptime)$-time
algorithm.

\paragraph{The binary search idea.}

We start by describing the idea when we have a rank oracle and then
discuss the independence oracle case. The first thing to note is that
one doesn't need the whole exchange graph for an augmentation. Instead,
what we need is to perform a breadth first search (BFS) on this ``unknown''
exchange graph and the rank oracle provides the following access
to this graph: for every vertex $a$ and every subset $B\subseteq\GS\setminus v$,
in $O(\log n)$ rank-oracle calls (by doing a binary-search style
argument) we can detect if $a$ has an edge to some vertex in $B$
or not. This suffices for doing a BFS on the graph in $O(n\log n\cdot\ranktime)$
time, and we do it at the beginning of each phase of Cunningham's
algorithm.

The second observation is that each phase in Cunningham's algorithm
can be implemented using $O(n\log n)$-many rank-oracle calls again.
Indeed, once the distances of vertices are known, shortest paths can
be computed using the aforementioned graph access which the rank-oracle
provides us. It is true that after some augmentations, some vertices
will be {\em misclassified} (or will be \emph{useless}, to borrow
Cunningham's~\cite{Cunningham-SICOMP86} terminology) but such vertices
are never queried again. The latter requires some distance-monotonicity
properties from~\cite{Cunningham-SICOMP86}. Theorem~\ref{thm:rank}
actually follows quite easily now: for an exact algorithm, there are
$O(\sqrt{r})$ such phases; for a $(1-\eps)$-approximate algorithm,
we need only $O(\frac{1}{\eps})$-phases. Details of this are given
in Section~\ref{sec:Rank-Oracle}.

Things are a bit trickier using only an independence oracle, since
we cannot detect, in $O(\log n)$ queries, whether a vertex $a$ has
an edge to a subset $B$ or not. Nevertheless, the following is true:
for a vertex $a\notin S$ (recall $S$ is the current solution) and
a subset $B\subseteq S$, in $O(\log r)$ independence queries we
can figure out whether $a$ has an edge to/from a vertex in $B$.
This is due to the way the exchange graph is defined. Armed with this
observation, after every augmentation, we can perform a BFS of the
new exchange graph (i.e., find the distance labels of all vertices)
in $O(n\log r)$ independence oracle calls \emph{plus} an $\tilde{O}(1)$
``amortized'' independence call per vertex whose distance label
changes in the current iteration. Since each vertex changes its distance
label at most $r$ times, the total amortized cost is $O(nr\cdot\indeptime)$.
Since there are at most $r$ augmentations, the total time taken in
the (non-amortized part of the) BFS computations is $O(nr\log r\cdot\Tind)$.
Details of this are described in Section~\ref{sec:Indep-Exact-Algorithm}.

\paragraph{Augmenting sets.}

Our exact algorithm with independence oracle queries has two kinds
of cost, both of which are $\tilde{O}(nr\cdot\indeptime)$-time. The
``amortized'' cost which is paid per ``distance increase'' per
vertex can be made $\tilde{O}(n/\eps)$ if we run only $1/\eps$ phases.
However, the ``non-amortized'' cost is paid per augmentation, and
this can still be $\tilde{O}(nr)$ even with $1/\eps$ phases. Indeed,
the phases don't seem to add any advantage. It seems unlikely that
the binary search idea alone can be salvaged to break the ``quadratic
barrier'' (when $r\approx n$).

To overcome this and prove Theorem$\,\ref{thm:apx}$, we propose the
idea of \emph{augmenting sets}. Recall the idea presented in Cunningham's
algorithm: the algorithm is akin to Hopcroft-Karp~\cite{HK-SICOMP73}
in that it runs in phases, and in each phase it augments along multiple
augmenting paths. However, the big difference is that the augmentations
need to be performed \emph{sequentially} rather than in parallel as
in Hopcroft-Karp, and this takes time. On the other hand, if we could
find all the augmentations that can occur \emph{up front}, then we
will save on the time taken to find them sequentially. This is what
our notion of augmenting sets achieves.

An augmenting set is a sequence of \emph{disjoint subsets} $B_{1},A_{1},B_{2},\ldots,A_{\ell},B_{\ell+1}$
where alternate subsets are not in, and respectively in, the current
solution $S$. More precisely, the set $A_{i}$ is a subset of vertices
at distance exactly $2i$ from the source of the exchange graph, and
vertices $B_{i}$ is at distance $2i-1$ from the source. Each subset
has the same cardinality and finally, deleting all the $A_{i}$'s
and adding all the $B_{i}$'s preserves independence in both matroids.
If the size of each set is $1$, an augmenting set is the same as
an augmenting path.

We prove an equivalence between augmenting sets and a collection of
augmenting paths which can be augmented sequentially in the exchange
graph. We introduce the concept of \emph{maximal} augmenting sets,
and show (a) as long as the shortest path is small, say $\ell$ (early
phases of the algorithm), the maximal and maximum augmenting sets
are within a multiplicative $O(\ell)$-factor, and (b) we show an
algorithm to find a ``near maximal'' augmenting set, which allows
us to guarantee that in any phase, and for any $p$, after spending
$O(n^{2}/p\cdot\indeptime)$ time the maximum number of remaining
$\ell$-length augmentations is $O(p\ell)$. These final augmentations
can be done in $\tilde{O}(np\ell\cdot\indeptime)$ time using the
previous binary search ideas. Setting $p\approx\sqrt{n}$ gives the
desired result (Theorem~\ref{thm:apx}). The precise definition of
augmenting sets, its properties, and the details of above ideas are
in Section~\ref{sec:Indep-Approximation-Algorithm}. We are hopeful
that this new class of algorithms, which may be of independent interest,
find further applications in related problems.

\paragraph{Frank-Wolfe sparsification and sampling.}

To obtain Theorem~\ref{thm:sparse}, we need one additional idea.
Since we are being approximate, we can further improve the running
time of our algorithm by sparsifying the ground set from $n$ to $\tilde{O}(r/\eps^{2})$.
We first look at the \emph{fractional} solution to the matroid intersection
problem as a convex optimization problem. Next, we observe that if
we apply a constrained first-order method (aka the Frank-Wolfe algorithm),
each step of the algorithm corresponds to solving the single matroid
optimization problem which can be done by the greedy algorithm in
$O(n\log n+n\cdot\Tind)$ time. Furthermore, to get an $\eps$-optimal
\emph{fractional} solution, one needs only $\tilde{O}(n/r\eps^{2})$
iterations which takes $\tilde{O}(\frac{n^{2}}{r\eps^{2}}\cdot\indeptime)$
time.

Next, we apply a sparsification procedure due to Karger~\cite{karger1998random}
that leads us to a ground set with only $\tilde{O}(r/\eps^{2})$ elements
(instead of $n$ elements), and which has a $(1-O(\eps))$-approximate
common independent set. Once the ground set shrinks, we can apply
the algorithm from the previous section (using augmenting sets) to
get a $(1-\eps)$-approximate solution in $\tilde{O}(r^{1.5}/\eps^{4.5}\cdot\indeptime)$
time. Taking everything together gives us Theorem~\ref{thm:sparse}.
The details are given in Section~\foreignlanguage{american}{\ref{sec:sample}}.

\subsection{Related Works}

Polynomial time algorithms for matroid intersection are more than
forty years old, with the first algorithms present in the works of
~\cite{Edmonds-Journal70,AigD71,Law75}. The running time of these
algorithms were $O(nr^{2}\cdot\Tind)$. Indeed, many of these papers~\cite{Law75,Fra81,BrezCG86,GabXu89,GabowXu-Journal96,ShiI95,LSW15,HuaKK16}
looked at the \emph{weighted} matroid intersection problem, and gave
polynomial time algorithms. 

For certain special matroids faster algorithms are known. Indeed,
when the matroid is given explicitly, one can talk of pure running
time instead of oracle queries. For instance, for exact maximum cardinality
bipartite matching, the best running time is the $O(m\sqrt{n})$-time
algorithm due to~\cite{HK-SICOMP73} and $\tilde{O}(m^{10/7})$-time
algorithm due to M\k{a}dry~\cite{Mad13}. Here $m$ is the number of
edges in the graph (and so, the number of elements in the matroid),
while $n$ is the number of vertices which is the rank of the matroid. 
%\sidford{Should we add something on approximate matching to satisfy the reviewers?}
%\snote{I suggest leaving citing Duan-Petie since that is  weighted setting which will confuse the reader here. If we want, we can  discuss somewhere else what's known in the weighted setting and that extending our results to wtd case is an interesting open problem.}
It is instructive to compare what our results give: we give an $\tilde{O}(m\sqrt{n}\cdot\Trank)$
and $\tilde{O}(mn\cdot\Tind)$ time algorithm. \textbf{}In dense
graphs, the best algorithm is an $O(n^{\omega})$-running time algorithm
by~\cite{MucS04,Har09}, where $\omega$ is the exponent of matrix
multiplication. For linear matroids, the matroid intersection problem
can be solved in essentially $O(nr^{\omega-1})$-time~\cite{Har09,CheLL14}.
For graphic matroids, the matroid intersection problem can be solved
in $O(n\sqrt{r}\log r)$ time~\cite{GabXu89,GabowXu-Journal96}.

The study of approximate matroid intersection problems is newer. As
Chekuri and Quanrud~\cite{ChekuriQuanrud-SODA16} note, Cunningham's
analysis implies a $O(nr/\eps\cdot\Tind)$-time algorithm to get an
$(1-\eps)$-approximate matroid intersection. Huang et al.~\cite{HuaKK16}
and Chekuri and Quanrud~\cite{ChekuriQuanrud-SODA16} study the approximate
\emph{weighted} version. The former gives an $\tilde{O}(nr^{1.5}/\eps\cdot\Tind)$-time
approximation algorithm~\cite{HuaKK16}, while the latter gives an
$O(nr\log^{2}(1/\eps)/\eps^{2}\cdot\Tind)$-time approximation algorithm.
Contrast this with our $\tilde{O}(nr\cdot\Tind)$-time exact and $\tilde{O}(n^{1.5}/\eps^{1.5})$-time
approximate algorithm, albeit for the unweighted version. Finally,
Guruganesh and Singla~\cite{GS-IPCO17} give an $\frac{1}{2}+\delta$-approximation
algorithm for a small but fixed constant $\delta$ which runs in $O(n\cdot\Tind)$-time.

We end the introduction with (to our knowledge) the only
known \emph{lower bound}  for matroid intersection. Since we are in
the oracle (rank/independence) model, it is perhaps foreseeable that
some non-trivial information theoretic lower bound can be attained
for the number of queries required for matroid intersection. Unfortunately,
the only lower bound we are aware of is due to Harvey~\cite{Har10,Harvey-SODA08}.
For matroids with $r=n/2$, Harvey~\cite{Har10} shows a lower bound
of $(\log_{2}3)n-o(n)$ queries. Obtaining an $\omega(n)$ lower bound
is a challenging open problem.

\section{Preliminaries}

Here we provide the notation (Section~\ref{sub:notation}) and previous
known results about the matroid intersection (Section~\ref{sub:matroid_lemmas})
that we use throughout the paper.

\subsection{Notation}

\label{sub:notation}

Here we provide the notational conventions we use throughout the paper.\\
\textbf{\vspace{-5pt}}\\
\textbf{Set Notation}: We often work with subsets of a finite set
$V$ which we call the \emph{universe} or \emph{ground set}. For $I\subseteq V$
and $a\in V$ we let $I+a\defeq I\cup\{a\}$ and $I-a\defeq I\setminus\{a\}$.
When the universe $V$ is clear from the context, we let $\overline{I}\defeq V\setminus I$.\textbf{
}We often abuse notation $A+B\defeq A\cup B$\textbf{ }for better readability.\textbf{}\\
\textbf{\vspace{-5pt}}\\
\textbf{Matroid}: A tuple $\M=(V,\I)$ for finite set $V$ and $\I\subseteq2^{S}$
is called a \emph{matroid }if (i) for every $A\in\I$ and $B\in\I$
where $|A|<|B|$, there exists an element $a\in B\setminus A$ such
that $A\cup a\in\I$, (ii) $\emptyset\in\I,$ and (iii) for every
$A'\subseteq A$ where $A\in\I$, we have $A'\in\I$. \textbf{}\\
\textbf{\vspace{-5pt}}\\
\textbf{Independent Sets}: We call $S\subseteq V$ \emph{independent}
if $S\in\I$ and \emph{dependent} otherwise. Further, for any $S\in\I$
we let $\free_{\M}(S)\defeq\{a\in\overline{S}\,|\,S+a\in\I\}$.\textbf{}\\
\textbf{\vspace{-5pt}}\\
\textbf{Matroid Polytope:} \foreignlanguage{american}{The matroid
polytope $\P$ is the convex hull of the indicator vectors of the
independent sets of $\M$.}\textbf{}\\
\textbf{\vspace{-5pt}}\\
\textbf{Matroid Intersection Polytope:} \foreignlanguage{american}{The
matroid intersection polytope is the convex hull of the indicator
vectors of the common independent sets of $\M_{1},\M_{2}$.}\textbf{
}It is well known that this polytope is given by \foreignlanguage{american}{$\P_{1}\cap\P_{2}$}.\textbf{}\\
\textbf{\vspace{-5pt}}\\
\textbf{Rank}: For a matroid $\M=(V,\I)$ we define the \emph{rank}
of $\M$ as $\rk(\M)=\max_{S\in\I}|S|$. Further, for any $S\subseteq V$
we define $\rk_{\M}(S)\defeq\max_{T\subseteq S|T\in I}|T|$, i.e.,
the size of the largest independent set contained in $S$. \textbf{}\\
\textbf{\vspace{-5pt}}\\
\textbf{Circuits}: For matroid $\M=(V,\I)$ we call $S\subseteq V$
a \emph{circuit} if it is a minimal dependent set, i.e., $S\notin\I$
but $S-a\in\I$ for any $a\in S$. For $S\in\I$ and $a\in\overline{S}$
such that $S+a\notin\I$ we let $\circ_{\M}(S+a)$ denote a minimal
dependent subset of $S+a$. \textbf{}\\
\textbf{}\\
\textbf{Exchangeable Pairs}: For a matroid $\M=(V,\I)$ and $S\in\I$
we call $a\in S$ and $b\notin S$ \emph{exchangeable} if $S-a+b\in\I$.
For a set $S\in\I$ we let $\exchange_{\M}(S)\defeq\{(a,b)\in S\times\overline{S}\,|\,S-a+b\in\I\}$
denote all exchangeable pairs in $S$.\textbf{}\\
\textbf{\vspace{-5pt}}\\
\textbf{Exchange Graph}: For matroids $\M_{1}=(V,\I_{1})$ and $\M_{2}=(V,\I_{2})$
and $S\in\I_{1}\cap\I_{2}$ we define the exchange graph $G(S)=(V\cup \{s,t\},E_{S})$
as the directed graph with vertices $V\cup\{s,t\}$ where $s$ and
$t$ are known as the \textit{source} and \textit{sink} vertices respectively.
There are the following 4 types of arcs in $E_S$:  (1) $(s,a)$ for all $a\in\free_{1}(S)$,
(2) $(a,t)$ for all $a\in\free_{2}(S)$, (3) $(a,b)$ for all $a\in S$
and $(a,b)\in\exchange_{1}(S)$, (4) $(b,a)$ for all $a\in S$ and
$(a,b)\in\exchange_{2}(S)$, where for $i\in\{1,2\}$ we use $\free_{i}$
and $\exchange_{i}$ as shorthand for $\free_{\M_{i}}$ and $\exchange_{\M_{i}}$
respectively.\textbf{}\\
\textbf{\vspace{-5pt}}\\
\textbf{Distances}: For any directed graph $G=(V,E)$ and $a,b\in V$
we let $d_{G}(a,b)$ denote the shortest path distance from $a$ to
$b$ using the edges of $E$. We define $d_{G}(a,a)=0$ for all $a\in V$
and if there is no $a$ to $b$ path then we let $d_{G}(a,b)=\infty$.\textbf{}\\
\textbf{\vspace{-5pt}}\\
\textbf{Distance} \textbf{sets}: For the exchange graph $G(S)$ we
denote the set of vertices at distance $i$ from source $s$ by $D_{i}$
or $L_{i}$.\foreignlanguage{american}{\textbf{}}\\
\textbf{\vspace{-5pt}}\foreignlanguage{american}{\textbf{}}\\
\foreignlanguage{american}{\textbf{Packing number:} For a matroid
$\M$ we let $\pac(\M)$ denote the maximum number of disjoint bases
in $\M$.}\textbf{}\\
\textbf{\vspace{-5pt}}\\
\textbf{Oracles}: Throughout this paper we assume that given a matroid
$\M=(V,\I)$ we can only access it through an oracle. We consider
two such oracle models. The first, is an \textbf{independence} oracle, which
when queried with any $S\subseteq V$, determines whether or not $S\in\I$
in time $\indeptime(\M)$. The second, is a \textbf{rank} oracle, which when
queried with an $S\subseteq V$, returns $\rk_{\M}(S)$ in time $\ranktime(\M)$.
Note that for $S\subseteq V$ we have that $S\in\I$ if and only if
$\rk_{\M}(S)=|S|$ and consequently $\indeptime(\M)=O(\ranktime(\M))$.
In the typical setting where we have two matroids $\M_{1}$ and $\M_{2}$
we let $\ranktime\defeq O(\ranktime(\M_{1})+\ranktime(\M_{2}))$ and
$\indeptime\defeq O(\indeptime(\M_{1})+\indeptime(\M_{2}))$.

\subsection{Matroid Theory }

\label{sub:matroid_lemmas}

Here we provide various standard results about the the structure of
matroids and matroid intersection used throughout the paper. The first
two lemmas are folklore.
\begin{lem}
\label{lem:circuits} For matroid $\M=(V,\I)$ and any $S\in\I$ and
$a\in V$ with $S+a\notin\I$, we have
\[
\circ_{\M}(S+a)=\left\{ b\in S+a\,|\,S+a-b\in\I\right\} \,.
\]
\end{lem}

\begin{lem}[Shortest augmenting paths]
\label{lem:basic_sap}Let $s,v_{1},...v_{a},t$ be a shortest path
from $s$ to $t$ in the exchange graph $G(S)$. Then $S+v_{1}-v_{2}+...-v_{a-1}+v_{a}\in\I_{1}\cap\I_{2}$,
i.e., augmenting along a shortest augmenting path preserves independence.
Sometimes we call this an \emph{augmentation} for brevity.
\end{lem}

The next three lemmas will be used extensively in our algorithms.
The first two show that to solve matroid intersection approximately,
it suffices to stop when the length of the shortest augmenting path
is $1/\epsilon$. The last lemma imposes control over the structure
of the exchange graph as we carry out augmentations.
\begin{lem}[Cunningham \cite{Cunningham-SICOMP86}]
\label{lem:augm-path-length} For any two matroids $\M_{1}=(V,\I_{1})$
and $\M_{2}=(V,\I_{2})$ with the largest common independent set of
size $r$, given a set $S\in\I_{1}\cap\I_{2}$ of size $|S|<r$,
there exists an augmenting path in $G(S)$ of size at most $1+2|S|/(r-|S|)$.
\end{lem}

\begin{cor}[\cite{Cunningham-SICOMP86,ChekuriQuanrud-SODA16,HuaKK16}]
\label{cor:approx-algo-certif} For any two matroids $\M_{1}=(V,\I_{1})$
and $\M_{2}=(V,\I_{2})$ with the largest common independent set of
size $r$, if the length of the shortest augmenting path in exchange
graph $G(S)$ for some $S\in\I_{1}\cap\I_{2}$ is at least $1/\epsilon$,
then $|S|\geq(1-O(\epsilon))\cdot r$.
\end{cor}

Our algorithms will rely on the following monotonicity lemma of Cunningham,
which says augmenting along the shortest path in an exchange graph
can only increase the distance of every element from the source and
sink vertices. 
\begin{lem}[Monotonicity Lemma\footnote{Throughout the paper, we will apply this lemma only on vertices $a$ which lie on a shortest $s,t$-path and therefore the conditions $d_{G(S)}(s,a) < d_{G(S)}(s,t)$ and $d_{G(S)}(a,t) < d_{G(S)}(s,t)$ hold. Furthermore, the second part shows that vertices with distance more than that cannot be part of an augmenting path of this length in the future. Cunningham's original result doesn't include this condition, and is false without it as corrected by \cite{Haselmayr,Price}, which is the first part of the lemma. Our version has an extra second part which is in fact also needed for proving the correctness of Cunningham's method. We thank Troy Lee for making us aware of these conditions.}, Cunningham \cite{Cunningham-SICOMP86}, Price~\cite{Price}, Haselmayr~\cite{Haselmayr}]
\label{lem:cunning-monot} For any two matroids $\M_{1}=(V,\I_{1})$
and $\M_{2}=(V,\I_{2})$ if we augment along the shortest path in
$G(S)$ obtaining $G(S')$ for a new set $S'\in\I_{1}\cap\I_{2}$
with $|S'|>|S|$. Denote $d = d_{G(S)}$ and $d' = d_{G(S')}$.  Then for all $a\in V$, 
\setenumerate{noitemsep,label=(\roman*)}
\begin{enumerate}
\item If $d(s,a)<d(s,t)$, then $d'(s,a)\geq d(s,a)$. Similarly, if $d(a,t)<d(s,t)$,
then $d'(a,t)\geq d(a,t)$. \label{cond1Cunning}
\item If $d(s,a)\geq d(s,t)$, then $d'(s,a)\geq d(s,t)$. Similarly, if
$d(a,t)\geq d(s,t)$, then $d'(a,t)\geq d(s,t)$. \label{cond2Cunning}
\end{enumerate}

%with $d_{G(S)}(s,a) < d_{G(S)}(s,t)$ and $d_{G(S)}(a,t) < d_{G(S)}(s,t)$, we have
%\[
%d_{G(S')}(s,a)\geq d_{G}(s,a)\text{~~~ and ~~~ }d_{G(S')}(a,t')\geq d_{G(S)}(a,t)\,.
%\]
\end{lem}

\begin{proof}
As $s$ and $t$ are symmetric, it suffices to prove the lemma
for $s$. Let $p$ be the augmenting path.
\ref{cond1Cunning} of the lemma follows from~\cite{Haselmayr,Price}.

For~\ref{cond2Cunning}, suppose that there is some $a$ for which $d(s,a)\geq d(s,t)$
and $d'(s,a)<d(s,t)$. Pick such $a$ with the minimal $d'(s,a)$.
Since $d(s,a)\geq d(s,t)$, $a\notin p$.

Consider the vertex $b$ right before $a$ on the shortest $s-a$
path in $G(S')$. We have $d'(s,b)=d'(s,a)-1<d(s,t)-2$. By the minimality
of $a$, $d(s,b)<d(s,t)$ so (1) applies to $b$. Thus $d(s,b)\leq d'(s,b)<d(s,t)-2$
which implies $(b,a)\notin G(S)$ (otherwise $d(s,a)<d(s,t)-1$).
We have two cases.

As exchanging $a,b$ in $S'$ preserves independence, and that $S$
itself is independent, $G(S)$ must contain an edge from $b$ or some
internal vertex on $p$ to $a$. This is a contradiction as $(b,a)\notin G(S)$.
\end{proof}

\selectlanguage{american}%
\global\long\def\findFree{\texttt{FindFree}}%
\global\long\def\findExchange{\texttt{FindExchange}}%
\global\long\def\outarc{\texttt{OutArc}}%
\global\long\def\getdistances{\texttt{GetDistancesRank}}%
\global\long\def\blockflow{\texttt{BlockFlow}}%

\global\long\def\M{\mathcal{M}}%
\global\long\def\P{\mathcal{P_{\M}}}%
\global\long\def\x{\textbf{\ensuremath{\mathbf{\mathbf{x}}}}}%
\global\long\def\z{\textbf{\ensuremath{\mathbf{z}}}}%
\global\long\def\I{\mathcal{\mathcal{I}}}%
\global\long\def\rk{\mathcal{\textsf{rank}}}%
\global\long\def\circ{\mathcal{\textsf{circuit}}}%
\global\long\def\free{\mathcal{\textsf{free}}}%
\global\long\def\otime{\mathcal{T}}%
\global\long\def\exchange{\textsf{exchange}}%

\selectlanguage{english}%

\section{Exchange Graph Exploration via Binary Search}

\label{sec:exchange_explore}

In this section we show how to use rank and independence oracles to
efficiently query the exchange graph for a pair of matroids, $\M_{1}=(V,\I_{1})$
and $\M_{2}=(V,\I_{2})$. These graph exploration primitives form
the basis of our exact and approximate matroid intersection algorithms
with a rank oracle as well as our exact matroid intersection algorithm
with an independence oracle.

In Section~\ref{sec:findfree} we show how to efficiently find free
vertices using a rank oracle and in Section~\ref{sec:findexchange}
we show how to efficiently find exchangeable pairs using an independence
oracle. Each routine works by a careful binary search with the appropriate
oracle and serves an important role in efficiently finding edges in
the exchange graph.

\subsection{Finding Free Vertices using Rank Oracle}

\label{sec:findfree}

We show that given a matroid $\M=(V,\I)$, an independent set $S\in\I$,
and a set $B\subseteq V$, we can efficiently find a free element
$e\in B$ with respect to $S$ (i.e., $S+e\in\I)$ or conclude that
there is none. The procedure $\findFree$ given in Algorithm~\ref{alg:findfree}
finds this element by binary search with the rank oracle. There is
a free element of $B$ if and only if $\rk_{\M}(B\cup S)>\rk(S)$
and consequently we can simply repeatedly divide $B$ in half, querying
if this property still holds. Formally, we analyze the performance
of this algorithm in Lemma~\ref{lem:findfree}.

\begin{algorithm2e}[H]

\label{alg:findfree}

\caption{$\findFree(\M=(V,\I),S\in\I,B\subseteq V)$}

\SetAlgoLined

\textbf{Input}: matroid $\M=(V,I)$, independent set $S\in\I$, and
subset $B\subseteq V$

\textbf{Output}: an element of $B\cap\free_{\M}(S)$ or $\emptyset$
if no such element exists

\lIf{ $\rk_{\M}(B\cup S)=\rk(S)$ }{ return $\emptyset$ }

\While{$|B|>1$ }{

Let $B_{1}$ and $B_{2}$ be an arbitrary partition of $B$ with $|B_{1}|\leq\lceil|B|/2\rceil$
and $|B_{2}|\leq\lceil|B|/2\rceil$

\lIf{$\rk_{\M}(B_{1}\cup S)>\rk_{\M}(B_{1})$ }{ $B\leftarrow B_{1}$ }

\lElse{ $B\leftarrow B_{2}$ }

}

\textbf{return:} the single element in $B$

\end{algorithm2e}
\begin{lem}[Finding Free Vertices]
 \label{lem:findfree} Given matroid $\M=(V,I)$, independent set
$S\in\I$, and elements $B\subseteq V$, the procedure $\findFree(\M,S,B)$
(Algorithm~\ref{alg:findfree}) outputs in $O(\log(|B|)\cdot\ranktime)$
time either an element of $B\cap\free_{\M}(S)$, or $\emptyset$ if
no such element exists.
\end{lem}

\begin{proof}
Note that $\free_{\M}(S)\cap B\neq\emptyset$ if and only if $\rk_{\M}(S\cup B)>\rk_{\M}(S)$.
Further, if $B_{1}$ and $B_{2}$ partition $B$ and $\rk_{\M}(S\cup B)>\rk_{\M}(S)$
then either $\rk_{\M}(S\cup B_{1})>\rk_{\M}(S)$ or $\rk_{\M}(S\cup B_{2})>\rk_{\M}(S)$.
Consequently, the output of the algorithm is correct. Since the size
of $|B|$ halves in each iteration of the while loop and the partitions
can be done arbitrarily, the running time is also as desired. 
\end{proof}
Given an independent set $S\in\I_{1}\cap\I_{2}$ for matroids $\M_{1}=(V,\I_{1})$
and $\M_{2}=(V,\I_{2})$ and an element $a\in S$, $\findFree$ gives
us an efficient way to find both incoming and outgoing arcs of $a$
in the exchange graph $G(S)$. Further, it gives an efficient way
to find arcs incident to $s$ and $t$ in the exchange graph. We leverage
this procedure for this purpose in Section~\ref{sec:Rank-Oracle}.

\subsection{Finding Exchange Vertices using Independence or Rank Oracle}

\label{sec:findexchange}

We show that given a matroid $\M=(V,\I)$, an independent set $S\in\I$,
an element $b\in\overline{S}$, and $A\subseteq S$ we can efficiently
find an element $a\in A$ that is exchangeable with $b$ (i.e., $S-a+b\in\I$)
or conclude that there is none. The procedure $\findExchange$, given
in Algorithm~\ref{alg:findexchange}, finds this element by binary
search using an independence oracle. We repeatedly divide the remaining
elements into two halves and figure out adding which half preserves
independence. Formally, we analyze the performance of this algorithm
in Lemma~\ref{lem:findfree}.

\begin{algorithm2e}[H]

\label{alg:findexchange}

\caption{$\findExchange(\M=(V,\I),S\in\I,b\in\overline{S},A\subseteq S)$}

\SetAlgoLined

\textbf{Input}: matroid $\M=(V,\I)$, independent set $S\in\I$, element
$b\in\overline{S}$, and a non-empty subset $A\subseteq S$

\textbf{Output}: $a\in A$ such that $(a,b)\in\exchange_{\M}(S)$
or $\emptyset$ if no such element exists

\lIf{ $S-A+b\notin\I$ }{ \textbf{return}: $\emptyset$ }

\While{$|A|>1$ }{

Let $A_{1}$ and $A_{2}$ be an arbitrary partition of $A$ with $|A_{1}|\leq\lceil|A|/2\rceil$
and $|A_{2}|\leq\lceil|A|/2\rceil$

\lIf{ $S-A_{1}+b\in\I$ }{ $A\leftarrow A_{1}$ }

\lElse{ $A\leftarrow A_{2}$ }

}

\textbf{return} the unique element of $A=\{a\}$

\end{algorithm2e}
\begin{lem}[Finding Exchange Vertices]
 \label{lem:exchange} Given a matroid $\M=(V,\I)$, independent
set $S\in\I$, element $b\in\overline{S}$, and a non-empty subset
$A\subseteq S$, the procedure $\findExchange(\M,S,b,A)$ (Algorithm~\ref{alg:findexchange})
outputs in $O(\log(|A|)\cdot\indeptime) = O(\log r\cdot \indeptime)$ time either an element
$a\in A$ such that $(a,b)\in\exchange_{\M}(S)$, or $\emptyset$
if no such element exists
\end{lem}

\begin{proof}
If $S-A+b\notin\I$ then of course no such $a$ exists. Assume $S-A+b\in\I$.

By considering $S\in\I$, we can add $|A|-1$ elements of $A$ to
$S-A+b$ while preserving independence. In particular, for any partition
of $A$ into nonempty $A_{1},A_{2}$, we must have $S-A_{1}+b\in\I$
or $S-A_{2}+b\in\I$. This proves the correctness of $\findExchange(\M,S,b,A)$
which just repeatedly performs this.

Since we halve the size of $A$ in each iteration, there can be at
most $O(\log|A|)$ iterations and the runtime follows.
\end{proof}
Given an independent set $S\in\I_{1}\cap\I_{2}$ for matroids $\M_{1}=(V,\I_{1})$
and $\M_{2}=(V,\I_{2})$ and an element $a\in\overline{S},$ $\findExchange$
gives us an efficient way to find both incoming and outgoing arcs
of $a$ in the exchange graph. We will leverage this procedure for
this purpose in Section~\ref{sec:Indep-Exact-Algorithm} and Section~\ref{sec:Indep-Approximation-Algorithm},
where we will exploit that this procedure only requires an independence
oracle.

In fact, the same procedure works for the rank oracle as well since
it easily implements the independence oracle:
\begin{lem}[Finding Exchange Vertices]
 \label{lem:exchange-1} Given matroid $\M=(V,\I)$, independent
set $S\in\I$, element $b\in\overline{S}$, and non-empty subset $A\subseteq S$,
the procedure $\findExchange(\M,S,b,A)$ (Algorithm~\ref{alg:findexchange})
outputs in $O(\log(|A|)\cdot\ranktime)= O(\log r\cdot \indeptime)$ time either $a\in A$
such that $(a,b)\in\exchange_{\M}(S)$, or $\emptyset$ if no such
element exists.
\end{lem}

\begin{proof}
A set $J$ is independent iff $\rk(J)=|J|$, so each independence
oracle call can be implemented by exactly one rank oracle call. Our
result then follows directly from Lemma \ref{lem:exchange}.
\end{proof}

\section{Exact and Approximation  Algorithms using Rank Oracle\label{sec:Rank-Oracle}}

In this section we consider the matroid intersection problem in the
rank oracle model. We assume throughout this section that $\M_{1}=(V,\I_{1})$
and $\M_{2}=(V,\I_{2})$ with $n\defeq|V|$ are given by rank oracles
and our goal is to provide both faster approximate and exact algorithms
for finding a maximum cardinality set in $\I_{1}\cap\I_{2}$. We split
the derivation of these algorithms into three parts.

In Section~\ref{sec:rkexplore} we provide basic primitives for exploring
the exchange graph with a rank oracle. Then, in Section~\ref{sec:rkblockflow}
we provide the main subroutine for our intersection algorithm which
efficiently computes a type of blocking flow in the exchange graph.
Finally, in Section~\ref{sec:ranksolve} we show how to use this
blocking flow subroutine to obtain our desired algorithms for solving
matroid intersection with a rank oracle.

\subsection{Exploring the Exchange Graph with a Rank Oracle}

\label{sec:rkexplore}

Here we show how to use the algorithms $\findFree$ (Algorithm~\ref{alg:findfree})
and $\findExchange$ (Algorithm~\ref{alg:findexchange}) of Section~\ref{sec:exchange_explore}
to efficiently find edges and distances in the exchange graph. First,
we provide the algorithm $\outarc$ (Algorithm~\ref{alg:outarc})
which, given any vertex $a$ and set $B$ of the exchange graph, quickly
finds an arc from $a$ to $B$ (if it exists). This algorithm simply
proceeds case by case finding the possible arcs. If $a=s$ then the
arcs can be found simply by looking for free vertices in $B$. If
$a\in S$ then the arcs can be found simply by looking for vertices in $B\cap \bar{S}$
that are free with respect to $S-a$. Similarly, if $a\notin S$
and $t\in B$ then edges to $t$ can be find just by checking if $S+a$
is independent. Finally, we can use $\findExchange$ to find arcs
from $a\in\bar{S}$ to $b\in S$. 

\begin{algorithm2e}[H]

\label{alg:outarc}

\caption{$\outarc$($S\in\I_{1}\cap\I_{2}$, $a\in V\cap\{s,t\}$, $B\subseteq V_{S}-a$)}

\SetAlgoLined

\textbf{Input}: independent set $S\in\I_{1}\cap\I_{2}$, vertex $a\in V_{S}$
of the exchange graph, and non-empty subset $B\subseteq V_{S}-a$
of the exchange graph. 

\textbf{Output}: $b\in B$ such that $(a,b)\in E_{S}$, or $\emptyset$
if no such element exists

\lIf{ $a=s$ }{ \textbf{return} $\findFree(\M_{1},S,B\setminus\{s,t\})$
}

\lIf{ $a\in S$ }{ \textbf{return} $\findFree(\M_{1},S-a,B\cap\bar{S})$
}

\lIf{ $a\in\bar{S}$, $t\in B$, and $S+a\in\I_{2}$}{ \textbf{return}
$t$ }

\lIf{ $a\in\bar{S}$ }{ \textbf{return} $\findExchange(\M_{2},S,a,B\cap S)$
}

\textbf{return:} $\emptyset$

\end{algorithm2e}
\begin{lem}[Finding Out Arcs Vertices]
 \label{lem:outarcs} Given independent set $S\in\I_{1}\cap\I_{2}$,
a vertex $a\in V_{S}$ of the exchange graph, and non-empty subset $B$, the
algorithm $\outarc(S,a,B)$ (Algorithm~\ref{alg:outarc}) outputs
in $O(\log(|B|)\cdot\ranktime)$ time either $b\in B$ such that $(a,b)\in E_{S}$,
or $\emptyset$ if no such element exists.
\end{lem}

\begin{proof}
The correctness and runtime follow almost immediately from Lemma~\ref{lem:findfree}
which provides guarantees on $\findFree$, and Lemma~\ref{lem:exchange-1}
which provides guarantees on $\findExchange$. To see the correctness,
simply consider the definition of the exchange graph; this algorithm
directly checks for the existence of each type of edge one at a time.
\end{proof}
Next, leveraging $\outarc$ (Algorithm~\ref{alg:outarc}) and its
analysis in Lemma~\ref{lem:outarcs} we show that we can efficiently
compute all distances from $s$ in the exchange graph. The
algorithm $\getdistances$ (Algorithm~\ref{alg:getdistances}) given below achieves
this goal by essentially performing  breadth first search (BFS) in
the exchange graph. For the sake of completeness, we include the details
below.

\begin{algorithm2e}[H]

\label{alg:getdistances}

\caption{$\getdistances$($S\in\I_{1}\cap\I_{2}$)}

\SetAlgoLined

\textbf{Input}: independent set $S\in\I_{1}\cap\I_{2}$. 

\textbf{Output}: $d\in\R^{V_{S}}$ such that for $a\in V\cap\{s,t\}$ we
have $d_{a}=d(s,a)$ the distance between $s$ and $a$ in $G(S)$.

Let $d_{a}=\infty$ for all $a\in V_{S}$

%Let $Q_{1}\leftarrows,Q_{-1}\leftarrow\emptyset$, $d_{s}=0$, and $B_{1}\leftarrowS\backslash s,B_{-1}\leftarrow\bar{S}$
Let $Q\leftarrow s$, and $B\leftarrow S\backslash s$

%\While{ $Q_{1}\cup Q_{-1}\neq\emptyset$ }{
\While{ $Q \neq\emptyset$ }{
%Let $a$ be the element added to $Q_{1}\cup Q_{-1}$ earliest. Suppose $a \in Q_i$, then $Q_{i}\leftarrowQ_{i}- a$.
Let $a$ be the element added to $Q$ earliest. 

%\While{ $b=\outarc(S,a,B_{-i})$ satisfies $b\neq\emptyset$ }{
\While{ $b=\outarc(S,a,B)$ satisfies $b\neq\emptyset$ }{

%Set $d_{b}=d_{a}+1$, $Q_{-i}\leftarrowQ_{-i}+b$, $B_{-i}\leftarrowB_{-i}-b$
Set $d_{b}=d_{a}+1$, $Q\leftarrow Q+b$, $B\leftarrow B-b$
}

}

\textbf{return:} $d$

\end{algorithm2e}
\begin{lem}[Finding Distances]
 \label{lem:finddistances} Given independent set $S\in\I_{1}\cap\I_{2}$
the $\getdistances(S)$ (Algorithm~\ref{alg:outarc}) outputs
in $O(n\log n\cdot\ranktime)$ time $d\in\R^{V_{S}}$ such that for
$a\in V_{S}$ we have $d_{a}=d(s,a)$ is the distance between $s$
and $a$ in the exchange graph $G(S)$.
\end{lem}

\begin{proof}
The procedure $\getdistances$ simply computes distances from $s$
using breadth first search. Note that each vertex is added to 
%$Q_{1}$ or $Q_{-1}$ 
$Q$
at most once and removed from 
%$B_{1}$ or $B_{-1}$ 
$B$
at most once. Consequently, the cost of the procedure is simply $O(n)$
plus the cost of $O(n)$ calls to $\outarc$. So the correctness and
runtime of the procedure follow from the analysis of $\outarc$ given
by Lemma~\ref{lem:outarcs}.
\end{proof}

\subsection{Blocking Flow with Rank Oracle}

\label{sec:rkblockflow}

In this section we provide our main DFS-like subroutine for improving
an independent set $S\in\I_{1}\cap\I_{2}$. This algorithm, $\blockflow$
(Algorithm~\ref{alg:blockflow}), efficiently performs an analog
of blocking flow for the exchange graph. This is essentially a fast
rank-oracle implementation of each phase in Cunningham's algorithm
\cite{Cunningham-SICOMP86}.

Given a set $S\in\I_{1}\cap\I_{2}$ it first computes the distance
from $s$ to every vertex in the exchange graph using $\getdistances$.
Using these distances, the algorithm subdivides $V$ into sets $L_{i}$
such that $L_{i}$ has all vertices at distance $i$ from $s$ in
the exchange graph of $S$. The algorithm then uses $\outarc$ to
look for a $s,a_{1},a_{2},...,a_{d_{t}-1},t$ path in $G(S)$ where
each $a_{i}\in L_{i}$. If found, the algorithm augments along such
a path removing those vertices and updating $S$ and looking for a
new path. However, whenever the algorithm concludes that some vertex
in a $L_{i}$ is not on such a path (i.e., a dead end), it removes
it from the graph. By carefully searching for these paths we show
that ultimately this algorithm is asymptotically no more expensive
then $\getdistances$ itself and will always find a larger set in
$\I_{1}\cap\I_{2}$ (if it exists) where the distance from $s$ to
$t$ in the exchange graph has increased.

\begin{algorithm2e}[H]

\label{alg:blockflow}

\caption{$\blockflow$($S\in\I_{1}\cap\I_{2}$)}

\SetAlgoLined

\textbf{Input}: independent set $S\in\I_{1}\cap\I_{2}$

\textbf{Output}: $S'\in\I_{1}\cap\I_{2}$ such that $|S'|>|S|$ and
$d_{G(S')}(s,t)\geq d_{G(S)}(s,t)+1$, or $S$ if no such $S'$ exists.

$d\leftarrow\getdistances$($S\in\I_{1}\cap\I_{2}$)

\lIf{ $d_{t}=\infty$ }{ \textbf{return} $S$ }

For all $i\in[d_{t}-1]$ let $L_{i}\leftarrow\{a\in V\,|\,d_{a}=i\}$

Let $L_{d_{t}}\leftarrow \{t\}$, $\ell \leftarrow 0$, $a_0 \leftarrow s$, and $L_{0} \leftarrow \{s\}$

\While{ $\ell\geq0$ }{

\If{ $\ell<d_{t}$ }{

\lIf{ $L_{\ell}=\emptyset$ }{ \textbf{return} $S$ }

$a_{\ell+1} \leftarrow \outarc(S,a_{\ell},L_{\ell+1})$ \tcp*{This takes $O(\ranktime \log n)$ time.}

\lIf{ $a_{\ell+1}=\emptyset$ }{ $L_{\ell}\leftarrow L_{\ell}-a_{\ell}$
and $\ell\leftarrow\ell-1$ }\tcp{This may set $\ell$ to $-1$ and is the only way the while loop exits.}

\lElse{ $\ell\leftarrow\ell+1$ }

}

\If{ $\ell=d_{t}$ }{

Augment along the path $s,a_{1},...,a_{d_{t}-1},t$ to obtain $S'\in\I_{1}\cap\I_{2}$
with $|S'|>|S|$

For all $i\in[d_{t}-1]$ set $L_{i}\leftarrow L_{i}-a_{i}$

$\ell\leftarrow 0$ and $S\leftarrow S'$

}

}
 
\textbf{return} $S$

\end{algorithm2e}
\begin{lem}[Blocking Flow]
 \label{lem:block_flow} Given independent set $S\in\I_{1}\cap\I_{2}$
the algorithm $\blockflow$ (Algorithm~\ref{alg:blockflow}) outputs
in $O(n\log n\cdot\ranktime)$ time $S'\in\I_{1}\cap\I_{2}$ such
that $|S'|>|S|$ and $d_{G(S')}(s,t)\geq d_{G(S)}(s,t)+1$, or $S$
if no such $S'$ exists (i.e., $S$ is the maximum cardinality set
in $\I_{1}\cap\I_{2}$).
\end{lem}

\begin{proof}
We first analyze the running time. By Lemma~\ref{lem:finddistances},
in $O(n\cdot\ranktime\log n)$ the algorithm $\getdistances$ will
ensure that $d_{a}=d_{G(S)}(s,a)$ for all $a\in V$. Once these distances
are computed the algorithm spends $O(n)$ time to subdivide $V$ into
the $L_{i}$.

Further, by Lemma~\ref{lem:outarcs} each invocation of $\outarc$
takes $O(\ranktime\log n)$ time and finds an edge from $a_{\ell}$
to $L_{\ell+1}$ if one exists. Now, in every iteration of the while
loop $\ell$ is either increased or decreased. Every time it is decreased,
either it is by a value of $1$, in which case $a_{\ell}$ is removed
from $L_{\ell}$ and there is no path of the form $s,a_{1},...,a_{d_{t}-1},t$
where each $a_{i}\in L_{i}$ and $a_{\ell}$ is included, or a path
is found and all the vertices on the path are removed. Consequently,
there are only $O(n)$ iterations of the while loop and the total
cost of the algorithm is as desired.

Now we show $S'\in\I_{1}\cap\I_{2}$. The reasoning in the preceding
paragraph shows that this algorithm simply finds a $s$ to $t$ path
through the $L_{i}$ if there is one, removing those vertices, and
repeating until there are no more paths. Such augmenting paths are
of length $d_{t}$ by design and hence they are in fact shortest augmenting
paths. This implies $S'\in\I_{1}\cap\I_{2}$ as augmenting along a
shortest augmenting path preserves independence (Lemma \ref{lem:basic_sap}).
Further there must be at least one augmentation since $\getdistances$
computed $d_{t}$. So $|S'|>|S|$.

It remains to show $d_{G(S')}(s,t)\geq d_{G(S)}(s,t)+1$. 
This is similar to the argument in Cunningham~\cite{Cunningham-SICOMP86}.
We argue
that there is no more augmenting path of length $d_{t}$ and consequently
the distance from $s$ to $t$ must increase. To this end, we must
prove that any elements removed from $L_{i}$ can no longer be on
any augmenting path of length $d_{t}$.

By Lemma~\ref{lem:cunning-monot}, distance of every vertex from
$s$ and to $t$ increases monotonically after each such augmentation.
%We will heavily rely on this fact. 
There are two ways a vertex is
removed. In the former case $a_{\ell}$ has no outgoing arc into $L_{\ell+1}$
and hence is a dead end. In the latter case they are on an augmenting
path of length $d_{t}$.

For the elements on such a path, since they enter or exit $S$ after
an augmentation, their distances from $s$ and to $t$ must strictly
increase. Consequently they cannot be on such a path anymore. Further,
since monotonicity implies that to stay on a $s$ to $t$ path of
distance $d_{G(S)}(s,t)$ for the original $G$ their distance from
$s$ cannot increase this implies that every vertex removed is not
on a $s$ to $t$ path of length $d_{t}$.
\end{proof}

\subsection{Matroid Intersection with a Rank Oracle}

\label{sec:ranksolve}

In this section we show how to use $\blockflow$ (Algorithm~\ref{alg:blockflow})
from the previous section to obtain our desired rank oracle based
algorithms for solving matroid intersection. All our algorithms simply
apply $\blockflow$ repeatedly to find the desired set. First, we
provide our result about exactly solving matroid intersection with
a rank oracle and then we provide our approximate result.

Both of our algorithms follow the structure of Cunningham's \cite{Cunningham-SICOMP86}.
The key difference is our fast subroutine $\blockflow$ for handling
each phase.
\begin{thm}[Exact Rank Oracle Algorithm]
\label{thm:exact_rank} Given matroids $\M_{1}=(V,\I_{1})$ and $\M_{2}=(V,\I_{2})$
there is an algorithm which finds the largest common independent set
of two matroids in $O(n\sqrt{r}\log n\cdot\ranktime)$  time  where $r$
is the size of the largest common independent set.
\end{thm}

\begin{proof}
Our algorithm simply starts with $S_{0}=\emptyset$ and iterates $S_{i+1}:=\blockflow(S_{i})$
until $S_{k+1}=S_{k}$ for some $k$ at which point the algorithm
outputs $k$. By Lemma~\ref{lem:block_flow} we have that $S_{i}\in\I_{1}\cap\I_{2}$
for all $\I$ and that $S_{k}$ is the desired largest common independent
set. Further, Lemma~\ref{lem:block_flow} implies that the running
time of this algorithm is $O(nk\ranktime\log n)$. Consequently, it
only remains to bound $k$.

Now, Lemma~\ref{lem:block_flow} also implies that for all $i<k$
we have $d_{G(S_{i+1})}(s,t)\geq d_{G(S_{i})}(s,t)+1$. So $d_{G(S_{i})}\geq i$
for all $i\in[k]$. Further, Corollary~\ref{cor:approx-algo-certif}
implies that $|S_{i}|\geq(1-O(1/i))r$. Consequently, $i=\Omega(\sqrt{r})$
implies that $|S_{i}|\geq r-O(\sqrt{r})$. Since every iteration of
blocking flow increases the size of $|S_{i}|$ this implies that increasing
$i$ by another $O(\sqrt{r})$ is enough to get a set in $\I_{1}\cap\I_{2}$
of size $r$. Consequently, $k=O(\sqrt{r})+O(\sqrt{r})=O(\sqrt{r})$.
\end{proof}
\begin{thm}[Approximate Rank Oracle Algorithm]
 \label{thm:approx-rank}Given matroids $\M_{1}=(V,\I_{1})$ and
$\M_{2}=(V,\I_{2})$ there is an algorithm which finds a $(1-\epsilon)$
approximation to the largest common independent set of the two matroids
in time  $O(n\epsilon^{-1}\log n\cdot\ranktime)$.
\end{thm}

\begin{proof}
Similar to proof of Theorem~\ref{thm:exact_rank} our algorithm simply
starts with $S_{0}=\emptyset$ and iterates $S_{i+1}:=\blockflow(S_{i})$.
However, here instead of repeating until $S_{k+1}=S_{k}$ we simply
output $S_{k}$ for $k=\Theta(\epsilon^{-1})$. Again, by Lemma~\ref{lem:block_flow}
we immediately have that the runtime is as desired and we have that
$S_{i}\in\I_{1}\cap\I_{2}$  and $d_{G(S_{i+1})}(s,t)\geq d_{G(S_{i})}(s,t)+1$
for all $i$. Since this implies $d_{G(S_{k})}=\Omega(\epsilon^{-1})$,
we have by Corollary$\ $\ref{cor:approx-algo-certif} that $|S_{k}|\geq(1-\epsilon)r$
where $r$ is the size of the largest common independent set of two
matroids and we have the desired result.
\end{proof}
In fact, both of our exact and approximate results hold even if we
have access to a rank oracle for one matroid and access to only an
independence oracle for the second. This is because $\findExchange$
requires only an independence oracle. We omit the detials.

\selectlanguage{american}%
\global\long\def\M{\mathcal{M}}%
 
\global\long\def\P{\mathcal{P_{\M}}}%
\global\long\def\x{\textbf{\ensuremath{\mathbf{\mathbf{x}}}}}%
\global\long\def\z{\textbf{\ensuremath{\mathbf{z}}}}%
\global\long\def\I{\mathcal{\mathcal{I}}}%
\global\long\def\rk{\mathcal{\textsf{rank}}}%
\global\long\def\circ{\mathcal{\textsf{circuit}}}%
\global\long\def\onepath{\texttt{OnePath}}%
\global\long\def\augpath{\texttt{AugmentingPaths}}%
\global\long\def\getdistancesIndep{\texttt{GetDistancesIndep}}%

\selectlanguage{english}%

\section{Exact Algorithm using Independence Oracle\label{sec:Indep-Exact-Algorithm}}

In this section we show how to solve matroid intersection exactly
using independence oracles. We assume throughout this section that
we are given matroids $\M_{1}=(V,\I_{1})$ and $\M_{2}=(V,\I_{2})$
with $n\defeq|V|$ and they can only be accessed via an independence
oracle. The goal is to compute the largest common independent set
and the following is our main result of this section. 
\begin{thm}
\label{thm:exact-indep}We can find the largest common independent
set of two matroids using $O(nr\log r\cdot\indeptime)$ time.
\end{thm}

To recall, the main challenge with independence oracles is that we do not have $O(\indeptime\log n)$ oracles to implement
the $\findFree$ subroutine. Nevertheless, we can obtain a speed-up using binary search ideas.

The run time of this algorithm can be broken into two components.
First, we show in Section$\ $\ref{subsec:indep-finding-layers} that
given any common independent set $S$, we can find the distance of
every element from $s$ in the exchange graph $G(S)$ in $O(n\log r\cdot\indeptime)$
time plus an ``amortized'' cost. This amortized cost will turn out
to be $O(\log r\cdot\indeptime)$ per element whenever its distance
increases from the sources. In our final $\augpath$ algorithm in
Section$\ $\ref{subsec:indep-augmenting-paths} we will exploit the
fact that this increase in distance can happen at most $O(r)$ times
for an element, so the overall amortized cost over all elements is
$O(nr\log r\cdot\indeptime)$. 

Second, given the distances in the exchange graph $G(S)$, in Section
\ref{subsec:Finding-One-Augmenting} we spend $O(n\log r\cdot\indeptime)$
time to find one augmenting path. Since there can be at most $r$
augmentations during the entire execution of the algorithm, this second
component will also be at most $O(nr\log r\cdot\indeptime)$.

\subsection{Finding the Distances\label{subsec:indep-finding-layers}}

We show how to compute all distances from $s$ in the exchange graph.
The following algorithm $\getdistancesIndep$ (Algorithm~\ref{alg:getdistances-indep})
achieves this goal simply by proceeding case by case. Starting with
the source $s$, it finds the distances by running a variant of BFS.
Going from an odd layer $\ell=2k+1$ to layer $\ell+1$ is easy as
we can just use the $\findExchange$ algorithm. To go from even layer
$\ell=2k$ to $\ell+1$, the idea is that we know up till $D_{\ell}$
where $D_{\ell}=\{a\in V|d_{a}=\ell\}$, and that we maintain some
candidates $L_{\ell+1}$ for $D_{\ell+1}$ (here $L_{\ell+1}$ are elements
we know are at distance at least $\ell+1$ from the source). Consider
any such candidate element $v\in L_{\ell+1}$. By one call to $\findExchange$
we can determine if $v$ has a neighbor in $D_{\ell}$. If yes, then
we know $v\in D_{\ell+1}$, otherwise it is at distance at least $\ell+3$,
so we move it to $L_{\ell+3}$.

\begin{algorithm2e}[H]

\label{alg:getdistances-indep}

\caption{$\getdistancesIndep$($S\in\I_{1}\cap\I_{2}$,$d'\in\R^{V_{S}}$)}

\SetAlgoLined

\textbf{Input}: independent set $S\in\I_{1}\cap\I_{2}$, lower bounds
$d'\in\R^{V_{S}}$ on distances from $s$ in $G(S)$.
One basic lower bound is $d'(s) = d'(t) = 0$, $d'(v) = 1$ for all $v\notin S$, and $d'(v) = 2$ for all $v\in S$.

\textbf{Output}: $d\in\R^{V_{S}}$ such that for $a\in V_{S}$ $d_{a}=d(s,a)$
is the distance between $s$ and $a$ in $G(S)$.

For all $i\in[d'_{t}-1]$ let $L_{i}\leftarrow \{a\in V\,|\,d'_{a}=i\}$

Let $d_{a}\leftarrow \infty$ for all $a\in V_{S}$ 

Let $d_{s}\leftarrow 0$, $L_{0} \leftarrow \{s\}$, and $\ell\leftarrow 0$

\While{ $\ell\geq0$ }{

\If{ $\ell$ is odd }{

\If{ $L_{\ell+1}=\emptyset$ }{ 

\lIf{$\{a\in V$ with $d_{a}\leq\ell\}=V$}{ \textbf{$d_{t}=\ell+1$
}}

\textbf{return} $d$ 

}

Let $Q \leftarrow L_{\ell+1}$

\While{ $L_{\ell}\text{ contains some }b$ }{

\While{ $a \leftarrow \findExchange(\M_{2},S,b,Q)$ satisfies $a\neq\emptyset$ }{

Set $d_{a} \leftarrow \ell+1$ and $Q \leftarrow Q-a$

}

$L_{\ell} \leftarrow L_{\ell}-b$

}

$L_{\ell+1} \leftarrow L_{\ell+1}-Q$ and $L_{\ell+3}\leftarrow L_{\ell+3}+Q$

}

\If{ $\ell$ is even }{

Let $Q \leftarrow L_{\ell+1}$

\While{ $Q\text{ contains some }b$ }{

\lIf{$a\leftarrow \findExchange(\M_{1}, S,b,L_{\ell})$ satisfies $a\neq\emptyset$
}{ Set $d_{b} \leftarrow \ell+1$ }

\lElse{ $L_{\ell+1} \leftarrow L_{\ell+1}-b$ and $L_{\ell+3} \leftarrow L_{\ell+3}+b$
}

$Q \leftarrow Q-b$

}

}

$\ell \leftarrow \ell+1$

}

\end{algorithm2e}
\begin{lem}[Finding Distances]
 \label{lem:Get-Distances-Indep} Given independent set $S\in\I_{1}\cap\I_{2}$
and distance lower bounds $d'\in\R^{V_{S}}$ for $G(S)$ , i.e., for
$a\in V_{S}$ we know $d'(s,a)\leq d_{G(S)}(s,a)$. The algorithm
$\getdistancesIndep(S,d')$ (Algorithm~\ref{alg:getdistances-indep})
outputs $d\in\R^{V_{S}}$ in $O(\sum_{a\in V}(1+d_{a}-d'_{a})\log r\cdot\indeptime)$
time such that for $a\in V_{S}$ we have $d_{a}=d_{G(S)}(s,a)$.
\end{lem}

\begin{proof}
To prove correctness of $\getdistancesIndep$, we prove the following
invariant at the beginning of any iteration of the outer while loop.

\begin{enumerate}

\item $D_{i}\subseteq\bigcup_{j\leq i}L_{j}$ for $i>\ell$ and $L_{\ell}=D_{\ell}$.

\item $d$ gives correct distances of elements at distance at most
$\ell$ from the source $s$.

\end{enumerate}

Clearly the invariant implies correctness of the algorithm when we
take $\ell=d_{G(S)}(s,t)$. The invariant is true for $\ell=0$ from
the definition of $d'$ and as we set $d_{s}=0$. To prove the invariant
for $\ell+1$, we separately consider $\ell$ being even or odd. 

When $\ell$ is odd the algorithm essentially performs BFS. It considers
vertices in $L_{\ell}$ one-by-one and check which of the vertices
in $L_{\ell+1}$ are reachable. All vertices that we can reach in
$L_{\ell+1}$ are definitely at distance $\ell+1$ and the remaining
vertices are at distance at least $\ell+3$.

When $\ell$ is even the algorithm considers every vertex in $L_{\ell+1}$
one-by-one and check if it is reachable from $L_{\ell}$. All vertices
that we can reach from $L_{\ell}$ are definitely at distance $\ell+1$,
and the remaining vertices are at distance at least $\ell+3$ so we
update $L_{\ell+3}$ accordingly.

To analyze the running time of the algorithm, consider any element
$a$. Suppose it was initially in $L_{i}$ according to $d'$ and
is finally at distance $j$ according to $d$. Observe that we change
the layers of this element at most $\lceil(j-i)/2\rceil$ times. In
each of these changes, we spend at most one call to $\findExchange$.
By Lemma$\ $\ref{lem:exchange} each such call takes $O(\indeptime\log r)$
time, which means the overall time taken by the algorithm is $O(\sum_{a\in V}(1+d_{a}-d'_{a})\log r\cdot\indeptime)$.
\end{proof}

\subsection{Finding One Augmenting Path\label{subsec:Finding-One-Augmenting}}

Next, given distances $d$ from $s$ to every vertex in the exchange
graph using $\getdistances$, we show how to find a single augmenting
path in $O(n\indeptime)$ time. Let $D_{i}=\{a\in V|d_{a}=i\}$.

The idea is to start with $\ell=d_{t}$ at the sink $t$ and perform
a depth-first search. To go from $\ell$ to $\ell-1$, i.e, given
a vertex $v\in D_{\ell}$, to find an adjacent vertex in $D_{\ell-1}$,
we try every possible vertex in $D_{\ell-1}$ and check if it has
an edge to $v.$ Notice that we are guaranteed to find some vertex
$u\in D_{\ell-1}$ with edge $(u,v)$ in $G(S)$ because $v\in D_{\ell}$.
By continuing this procedure, we are guaranteed to find an augmenting
path of length $d_{t}$. Since we reach or check each vertex at most
once, the overall running time of finding this augmenting path will
be $O(n)$.

\begin{algorithm2e}[H]

\label{alg:one-Aug-indep}

\caption{$\onepath$($S\in\I_{1}\cap\I_{2}$,$d\in\R^{V_{S}}$)}

\SetAlgoLined

\textbf{Input}: independent set $S\in\I_{1}\cap\I_{2}$, $d\in\R^{V_{S}}$
such that for $a\in V_{S}$ $d_{a}=d(s,a)$ is the distance between
$s$ and $a$ in $G(S)$. 

\textbf{Output}: $S'\in\I_{1}\cap\I_{2}$ after performing one augmentation
along shortest path in $G(S)$, or $S$ if no such augmentation exists.

For all $i\in[d_{t}-1]$ let $L_{i}\leftarrow \{a\in V\,|\,d_{a}=i\}$

\lIf{ $d_{t}=\infty$ }{ \textbf{return} $S$ }

Let $\ell \leftarrow d_{t}$ and $a_{d_{t}} \leftarrow t$

\While{ $\ell  > 0$ }{

Let $Q \leftarrow L_{\ell-1}$

\While{ $Q\neq\emptyset$ }{

Pick arbitrary $b\in Q$

Set $j\leftarrow 1$ if $\ell$ is even abd $j\leftarrow 2$ if $\ell$ is odd.

\lIf{ $S-a_{\ell}+b\in\I_{j}$ }{Set $a_{\ell-1} \leftarrow b$ and \textbf{break}}

\lElse{$Q\leftarrow Q-b$}

}

$\ell \leftarrow \ell-1$

}

Augment along the path $s,a_{1},...,a_{d_{t}-1},t$ to obtain $S'\in\I_{1}\cap\I_{2}$
with $|S'|>|S|$

\textbf{return} $S'$

\end{algorithm2e}
\begin{lem}[One Augmentation]
 \label{lem:one-aug-Indep} Given independent set $S\in\I_{1}\cap\I_{2}$
and distance $d\in\R^{V_{S}}$ such that for $a\in V_{S}$, $d_{a}=d(s,a)$
is the distance between $s$ and $a$ in the exchange graph $G(S)$,
the algorithm $\onepath(S,d)$ (Algorithm~\ref{alg:one-Aug-indep})
outputs $S'\in\I_{1}\cap\I_{2}$ with $|S'|>|S|$ in $O(n \cdot \indeptime)$
time by performing an augmentation along a shortest path in $G(S)$,
or $S$ if there is no such $S'$.
\end{lem}

\begin{proof}
Starting with the sink vertex $t$, the algorithm iteratively finds
an element $a_{\ell-1}$ at distance $\ell-1$ from $s$ such that
there is an edge $(a_{\ell-1},a_{\ell})$. It tries every vertex $b\in L_{\ell-1}$
and check if there is an edge $(b,a)$. We are guaranteed to find
some $b$ because $a_{\ell}$ is at distance $\ell$ from $s$. The
time taken by this step is $O(|L_{\ell-1}|\cdot\indeptime)$. Since
$\sum_{\ell}|L_{\ell-1}|\leq n$, the overall time taken by the algorithm
is $O(n\cdot\indeptime)$.
\end{proof}

\subsection{The Augmenting Paths Algorithm\label{subsec:indep-augmenting-paths}}

Now we use the subroutines from the last section to find the largest
common independent set in any two given matroids and prove Theorem$\ $\ref{thm:exact-indep}.
The $\augpath$ algorithms alternates between $\getdistancesIndep$
and $\onepath$ algorithms to perform augmentations.

\begin{algorithm2e}[H]

\label{alg:one-Aug-indep-1}

\caption{$\augpath$($\M_{1},\M_{2}$)}

\SetAlgoLined

\textbf{Input}: Two matroids $\M_{1}=(V,\I_{1})$ and $\M_{2}=(V,\I_{2})$. 

\textbf{Output}: $S\in\I_{1}\cap\I_{2}$ of maximum size

Let $S\in\I_{1}\cap\I_{2}$ be a maximal solution greedily found

Let $d\in\R^{V_{S}}$ be defined by $d_{s}=0$, $d_{t}=1$, $d_{v}=2$
for $v\in S$, and $d_{v}=1$ for $v\in\overline{S}$.

\While{$\mathrm{true}$}{

$d \leftarrow \getdistancesIndep(S,d)$

\lIf{ $d_{t}\neq\infty$ }{$S \leftarrow \onepath(S,d)$ }

\lElse{\textbf{return} $S$}

}

\end{algorithm2e}
\begin{proof}[Proof of Theorem$\ $\ref{thm:exact-indep}]
 Starting with a feasible solution $S$, since the $\augpath$ algorithm
finds all shortest augmenting paths one-by-one, it clearly returns
the optimal solution. 

Next we bound the running time. The maximal common independent set
$S\in\I_{1}\cap\I_{2}$ can be greedily computed by adding an element
$e$ to $S$ if $S+e\in\I_{1}\cap\I_{2}$. This only takes $O(n\cdot\indeptime)$
time. To bound the overall time taken by all calls to $\getdistancesIndep$
we use Lemma$\ $\ref{lem:Get-Distances-Indep}. Notice that an element
$a$'s distance $d_{a}$ only increases during the execution of the
algorithm. Since its final distance is at most $r$, we get the overall
time is $O(nr\log r\cdot\indeptime)$. Finally, each call to $\onepath$
takes $O(n\indeptime)$ time by Lemma$\ $\ref{lem:one-aug-Indep}.
Since the total number of augmentations performed using $\onepath$
is $O(r)$, these calls take $O(nr\cdot\indeptime)$ time.
\end{proof}
\begin{rem*}
In fact we can achieve a slightly better runtime of $O((nr+r^{2}\log r)\cdot\indeptime)$
time. When computing the distances from an even layer $L_{2k}$ to
an odd layer $L_{2k+1}$, for each $b\in L_{2k+1}$ we check if $b$
has an incoming edge from $L_{2k}$ using $\findExchange$ which takes
$O(\log r\cdot\indeptime)$ time. Instead of explicitly finding such
an edge, we can simply check if $S-L_{2k}+b$ is independent and determine
if $b$ has distance $2k+1$. This requires only one independence
call. We presented the slightly slower implementation because it is
easier to understand.
\end{rem*}

\begin{rem*}
Nguy\~{\^e}n's recent concurrent and independent note~\cite{nguyen2019arxiv} exploits a similar binary search idea and obtains an algorithm with $O(nr\log^2 r)$ independence queries. The reason for the extra $\log r$ factor seems to arise because he does not amortize the cost per vertex as we do in Lemma~\ref{lem:Get-Distances-Indep}. Indeed, if we do not have the $d'_a$ term, then every $a\in V$ would pay $O(d_a\log r \cdot \indeptime)$ and it is known (can be derived from Lemma~\ref{lem:augm-path-length}) that $\sum_a d_a = O(r\log r)$. This would give Nguy\~{\^e}n's factor~\cite{nguyen2019arxiv}.
\end{rem*}

\section{Approximation Algorithm using Independence Oracle\label{sec:Indep-Approximation-Algorithm}}

In this section we consider the problem of obtaining a $(1-\epsilon)$-approximate
common independent set in two matroids $\M_{1}=(V,\I_{1})$ and $\M_{2}=(V,\I_{2})$
using independence oracles. Our main result is the following theorem. 
\begin{thm}
\label{thm:approxMatrInters} There is an $O(\frac{n^{1.5}\sqrt{\log r}}{\eps^{1.5}}\cdot\indeptime)$-time
algorithm to obtain a $(1-\eps)$-approximation to the matroid intersection
problem.
\end{thm}

\paragraph{Overview.}

Our algorithm uses Cunningham's idea of performing augmentations in
phases. Let $S$ denote the solution of the algorithm at the start
of a phase. Denote the exchange graph of $S$ by $G(S)$. Let the
length of the shortest path in $G(S)$ be $2(\ell+1)$. Let $D_{k}$
be the set of elements at distance $k\leq2(\ell+1)$ from the source
$s$. Note by Lemma~\ref{lem:Get-Distances-Indep}, we can compute
these sets in $O(n\log r\cdot\Tind)$ time. The rough idea of Cunningham's
algorithm in a phase is to perform a maximal collection of augmentations
as long as the shortest path length remains the same. Once such a
maximal collection is augmented, the distance from source to sink
increase by at least 2 and we move on to the next phase. Note that, by Corollary~\ref{cor:approx-algo-certif},
one obtains a $(1-\epsilon)$-approximate solution after $O(1/\epsilon)$
phases. 
%This is because it ensures the shortest augmenting path is
%of length $\Omega(1/\epsilon)$, which by a simple counting argument
%implies a $(1-\epsilon)$-approximate solution. 
Thus the total running
time of the algorithm is $O(1/\epsilon)$ times the \textit{time}
(number of independence oracles) needed to process each phase. The
rest of this section therefore focuses on how to process a phase fast.

In each phase, Cunningham's algorithm takes one augmenting path at
a time. This becomes necessary for any ``augmenting paths'' style
algorithm because each augmentation might add or delete several edges
in the exchange graph. For example, taking one augmenting path may
destroy another {\em disjoint} augmenting path\footnote{As a concrete example, consider the maximum bipartite matching problem
(which is an intersection of two partition matroids) on the graph
being a path $\{e_{1},e_{2},\ldots,e_{6}\}$ of length 6. Suppose
the current matching is $\{e_{2},e_{5}\}$. The corresponding exchange
graph has two disjoint augmenting paths $\{e_{1},e_{2},e_{3}\}$ and
$\{e_{4},e_{5},e_{6}\}$. However, if we take both the augmenting
paths simultaneously then the obtained set $\{e_{1},e_{3},e_{4},e_{6}\}$
is not a valid matching.}. Unfortunately,
this means we cannot beat $O(nr)$ since we don't know how to do an
augmentation in better than $O(n)$ time. To overcome this barrier,
our main idea is to simultaneously find multiple augmenting paths.
Of course this needs care as one augmenting path might destroy another
path. Our crucial observation is that the union of a sequence of consecutive
shortest augmenting paths can be interpreted as ``augmentation sets''.
These augmentation sets satisfy many of the properties of an augmenting
path. Although finding a maximal collection of augmenting sets still
takes quadratic time, to obtain a sub-quadratic algorithm we halt
this Augmenting Sets algorithm early in Section~\ref{subsec:FasterAlgo}
and combine it with our Augmenting Path algorithm from Section~\ref{sec:Indep-Exact-Algorithm}.

Before describing our Augmenting Sets algorithm in Section~\ref{subsec:TheAugSetsAlgo},
we start by proving some basic facts about matroids and then studying 
the properties of \textit{augmenting sets}.

\subsection{Facts about Matroids}
\begin{fact}
\label{fact:1} Let $S$ be a common independent set in two matroids
$\M_{1}=(V,\I_{1})$ and $\M_{2}=(V,\I_{2})$. Let $G(S)$ be the
exchange graph w.r.t $S$. If there exists a set $A\subseteq S$ and
$b\notin S$ such that there is no edge of the form $(a,b)\in E(G(S))$
for any $a\in A$, then $S-A+b\notin\I_{1}$. Similarly, if there
is no edge of the form $(b,a)\in E(G(S))$ for any $a\in A$, then
$S-A+b\notin\I_{2}$. 
\end{fact}

\begin{proof}
Direct from the defintion of matroids.
\end{proof}

%\begin{proof}
%Let us do the first part; the other is analogous. \smallskip{}
%
%\noindent $(a,b)\notin E(G(S))$ implies $S-a+b\notin\I_{1}$. In
%particular, $S+b\notin\I_{1}$. Since $S$ is independent, $S+b$
%has a unique minimal circuit $C$ (dependent set) defined as $C=\{x\in S:S-x+b\in\I_{1}\}$
%(cite prelim:lemma). Since $S-a+b\notin\I_{1}$ for all $a\in A$,
%this shows $A\cap C=\emptyset$. In particular, $S-A+b$ contains
%the circuit $C$ and thus $S-A+b\notin\I_{1}$. \smallskip{}
%
%A self-contained rank-based proof by induction. Let $A=\{a_{1},a_{2},\ldots,a_{k}\}$.
%Let's prove $\rk_{1}(S-\{a_{1},a_{2}\}+b)=|S|-2<|S|-1$ proving $S-\{a_{1},a_{2}\}+b\notin\I_{1}$.
%The full proof is to inductively apply this. From submodularity of
%rank, 
%\[
%\rk_{1}(S-a_{1}+b)+\rk_{1}(S-a_{2}+b)\geq\rk_{1}(S+b)+\rk_{1}(S-\{a_{1},a_{2}\}+b)
%\]
%The LHS is $2(|S|-1)$ since $S-a_{i}+b\notin\I_{1}$. $\rk_{1}(S+b)=|S|$
%as well ($S+b\notin\I_{1}$ but $S\in\I_{1}$). Thus, $\rk_{1}(S-\{a_{1},a_{2}\}+b)=|S|-2$. 
%\end{proof}
\begin{fact}
\label{fact:2} Let $\M=(V,\I)$ be any matroid and let $S\in\I$
be an independent set. Let $X\subseteq S$ and $Y,Z\subseteq\bar{S}$
with $Y\cap Z=\emptyset$ satisfy the following properties:\setenumerate{label=(\alph*),noitemsep}
\begin{enumerate}
\item $S+Z\in\I$
\item $S-X+Y\in\I$
\item $Y\in\Span(S)$, that is, $\rk(S+Y)=\rk(S)=|S|$.
\end{enumerate}
Then, $S+Z-X+Y\in\I$.
\end{fact}

\begin{proof}
By submodularity of rank, 
\[
\rk(S-X+Y+Z)+\rk(S+Y)\geq\rk(S-X+Y)+\rk(S+Y+Z).
\]
Now, $\rk(S+Y)=|S|$ (given as (c)), $\rk(S-X+Y)=|S|-|X|+|Y|$ (given
as (b)), and $\rk(S+Y+Z)\geq\rk(S+Z)=|S|+|Z|$. Thus, $\rk(S-X+Y+Z)=|S|-|X|+|Y|+|Z|$. 
\end{proof}
%\begin{fact}\label{fact:3}
%	Let $\M = (V,\I)$ be a matroid and $S\subseteq V$. Let $P \subseteq A\subseteq S$ and $Q\subseteq B \subseteq \bar{S}$ be such that 
%	\begin{enumerate}[noitemsep, label=(\alph*)]
%		\item $S - A + B\in \I$
%		\item $S - P + Q \in \I$
%		\item $|P| = |Q|$ but $|A| \geq |B|$.
%	\end{enumerate}
%	Then, one can find a set $R$ such that $P \subseteq R \subseteq A$ with $|R| = |B|$ such that $S - R + B \in \I$.
%	%
%	Furthermore, this can be found in $|A| - |P|$ independence oracle queries. 
%\end{fact}
%\begin{proof}
%	Almost trivially follows from the matroid property. 
%	If $|A| = |B|$, then $R = A$. Therefore, assume $|A| > |B|$ implying $|S - A + B| < |S|$ while $|S - P + Q| = |S|$.
%	Thus, there is an element in $(S - P + Q)\setminus (S - A + B) = A\setminus P$ (since $Q\subseteq B$) which can be added to $S - A + B$ while maintaining independence.
%	We keep adding these elements from $A\setminus P$ to $R$ (which is initialized to $P$) till $|S - R + B| = |S|$.
%	We need to only ask independence oracle queries with elements in $A\setminus P$, and thus the running time follows.
%\end{proof}

\subsection{Augmenting Sets}

Let $S$ be a common independent set in $\I_{1}\cap\I_{2}$. Recall
the exchange graph $G(S)$ with respect to this common set $S$, whose
vertex set $V(G(S))=V\cup\{s,t\}$. Also recall that $D_{i}$ is the
set of elements of $\GS$ which are at distance \emph{exactly} $i$
from the source $s\in G(S)$. Next, we define the notion of augmenting
sets in the graph $G(S)$. Let $2(\ell+1)$ be the length of the shortest
path from $s$ to $t$ in the augmenting graph. We define augmenting
sets with respect to this graph $G(S)$.

%. which will be shown to
%be equivalent to augmenting a sequence of (shortest) augmenting paths
%of the same length. By using augmenting sets, a sequence of augmenting
%paths can be found faster. Readers may contrast this with the famous
%Hopcroft-Karp\textasciitilde\textbackslash cite\{\} algorithm which
%finds a disjoint collection of shortest augmenting paths. For bipartite
%matching, we can augment along them as long as the paths are disjoint;
%however, for matroids we must take into account the interaction between
%elements on the augmenting paths. 
\begin{defn}[Augmenting Sets]
 Let $S\in\I_{1}\cap\I_{2}$ and $G(S)$ be the corresponding exchange
graph with shortest path $2(\ell+1)$. A collection of sets $\Pi_{\ell}:=(B_{1},A_{1},B_{2},A_{2},\ldots,A_{\ell},B_{\ell+1})$
%Sets $\{A_{k},B_{k}\}_{k}$ 
form an augmenting collection of sets, or simply augmenting sets,
in $G(S)$ if the following conditions are satisfied: %$B_{0}\subseteq D_{0},A_{1}\subseteq D_{1},B_{1}\subseteq D_{2},...,B_{l}\subseteq D_{2l}$
%satisfy the following:
\global\long\def\width{w}%
\setenumerate{label=(\alph*),noitemsep}
\begin{enumerate}
\item For $1\leq k\leq\ell+1$, we have $A_{k}\subseteq D_{2k}$ and $B_{k}\subseteq D_{2k-1}$.
\item $|B_{1}|=|A_{1}|=|B_{2}|=\cdots=|B_{\ell+1}|=\width$
\item $S+B_{1}\in\mathcal{I}_{1}$
\item $S+B_{\ell+1}\in\mathcal{I}_{2}$
\item For all $1\leq k\leq\ell$, we have $S-A_{k}+B_{k+1}\in\mathcal{I}_{1}$
\item For all $1\leq k\leq\ell$, we have $S-A_{k}+B_{k}\in\mathcal{I}_{2}$
\end{enumerate}
See Figure~\ref{fig:augSets} for an illustration. We let $\width$ denote the \emph{width} of $\Pi_{\ell}$. 
\end{defn}

There is some redundancy in the above definition. For instance, $S+B_{1}\in\I_{1}$
does imply $B_{1}\subseteq D_{1}$, the set of free vertices for $S$
in $\I_{1}$. Nevertheless, we include this in the definition for
better understanding. Observe that if each $|A_{k}|$ was of size $1$,
an
augmenting set corresponds to an augmenting path. Thus, the
concept generalizes the concept of augmenting paths in the exchange
graph and in the next subsection we prove many analogous properties.

%We say augmenting sets $\{A_{k},B_{k}\}_{k}$ \textit{contain} augmenting
%sets $\{A'_{k},B'_{k}\}_{k}$ if for all $k$ we have $A'_{k}\subseteq A_{k}$
%and $B_{k}'\subseteq B_{k}$. Observe that the definition of augmenting
%sets is identical to that of augmenting paths if $A_{k},B_{k}$ are
%single elements instead of sets. We shall show that many of the properties
%of augmenting paths carry over to augmenting sets.

%\documentclass[11pt]{article}

\usetikzlibrary{decorations.markings}
\usetikzlibrary{arrows}
\tikzstyle{graphnode}=[circle, draw, fill=black!20, inner sep=0pt, minimum width=3pt]

\tikzset{
    %Define standard arrow tip
    >=stealth',
    % Define arrow style
    pil/.style={
           ->,
           thick,
           shorten <=2pt,
           shorten >=2pt,}
}
\tikzset{->-/.style={decoration={
  markings,
  mark=at position .5 with {\arrow{>}}},postaction={decorate}}}
%} %MAC

%\begin{document}

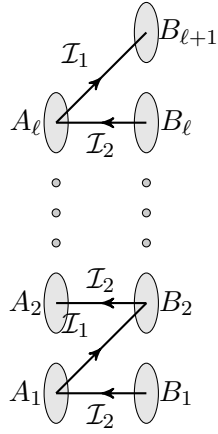
\begin{figure}
	\centering
	\begin{tikzpicture} [thin,scale=0.8]
		\draw (0,-1.5) ellipse (0.2cm and 0.5cm)[fill=black!10];
		\draw (0,0) ellipse (0.2cm and 0.50cm)[fill=black!10];
%		\draw (0,1.5) ellipse (0.2cm and 0.50cm)[fill=black!10];
		\draw (0,3) ellipse (0.2cm and 0.50cm)[fill=black!10];

		\draw (1.5,-1.5) ellipse (0.2cm and 0.50cm)[fill=black!10];
		\draw (1.5,0) ellipse (0.2cm and 0.5cm)[fill=black!10];
%		\draw (1.5,1.5) ellipse (0.2cm and 0.50cm)[fill=black!10];
		\draw (1.5,3) ellipse (0.2cm and 0.50cm)[fill=black!10];
		\draw (1.5,4.5) ellipse (0.2cm and 0.50cm)[fill=black!10];

		\draw [->-,   thick] (1.5,-1.5) to node[auto]{$\mathcal{I}_2$} (0,-1.5);				
		\draw [->-,   thick] (0,-1.5) to node[auto]{$\mathcal{I}_1$} (1.5,0);
		\draw [->-,   thick] (1.5,0) to node[auto,label=above:$\mathcal{I}_2$]{} (0,0);				
%		\draw [->-,   thick] (0,0) to node[auto,label=above:$\mathcal{I}_1$]{$ $} (1.5,1.5);
%
		\draw [->-,   thick] (1.5,3) to node[auto]{$ \mathcal{I}_2$} (0,3);				
		\draw [->-,   thick] (0,3) to node[auto]{$\mathcal{I}_1$} (1.5,4.5);

		\node at (1.5 , 1) [graphnode]{};
		\node at (1.5 , 1.5) [graphnode]{};		
		\node at (1.5 , 2) [graphnode]{};

		\node at (0 , 1) [graphnode]{};
		\node at (0 , 1.5) [graphnode]{};		
		\node at (0 , 2) [graphnode]{};

		\node at (-0.5,-1.5) {$A_1$};		
		\node at (2,-1.5) {$B_1$};
		\node at (-0.5,0) {$A_2$};
		\node at (2,0) {$B_2$};
		
		\node at (-0.5,3) {$A_{\ell}$};
		\node at (2,3) {$B_{\ell}$};
		\node at (2.2,4.5) {$B_{\ell+1}$};

	\end{tikzpicture}
    \caption[margin=5cm]{Augmenting sets where each $A_k \subseteq S$ and $B_k \subseteq \overline{S}$. They satisfy $S - A_k + B_{k+1} \in \mathcal{I}_1$ and $S - A_k + B_{k} \in \mathcal{I}_2$.} 
        \label{fig:augSets}
\end{figure}

%\end{document}

\subsubsection{Properties of augmenting sets}

Our first goal is to show that given an augmenting set $\Pi_{\ell}$,
we can do the natural swap operation to obtain a larger common independent
set. 
\begin{thm}
\label{thm:aug-means-aug} Let $\Pi_{\ell}:=(B_{1},A_{1},B_{2},A_{2},\cdots,B_{\ell},A_{\ell},B_{\ell+1})$
be the an augmenting set in the exchange graph $G(S)$ whose shortest
path length is $2(\ell+1)$. %	Let $2l$ be the length of the shortest augmenting path and $\{A_{k},B_{k}\}_{k}$
%	form augmenting sets. 
 Then the set $S':=S\oplus\Pi_{\ell}:=S+B_{1}-A_{1}+B_{2}-\cdots+B_{\ell}-A_{\ell}+B_{\ell+1}$
is a common independent set. 
\end{thm}

\begin{proof}
Let us prove that $S'\in\I_{1}$, and the proof of $S'\in\I_{2}$
is analogous and omitted.

The first observation is that there is no edge from a vertex $a\in A_{i}$
to a vertex $b\in B_{j}$ with $j>i+1$. The reason is that $A_{i}\subseteq D_{2i}$
and $B_{j}\subseteq D_{2j-1}$. If $j>i+1$, then $2j-1>2i+1$, and
since the $D_{i}$'s are vertices at distance exactly $i$, there
can't be an edge from any vertex in $D_{2i}$ to any vertex in $D_{2j-1}$.

First, in the following lemma, we prove $S'-B_{1}$ is independent in $\M_1$.
\begin{lem}
	\label{lem:augset} The set $S'-B_{1}:=S-A_{1}+B_{2}-A_{2}+\cdots+B_{\ell+1}\in \I_{1}$. 
\end{lem}
\begin{proof}
	\global\long\def\calA{\mathcal{A}}%
	
	\global\long\def\calB{\mathcal{B}}%
	For brevity, let us use $\calA_{k}$ to denote the union of the sets
	$A_{1}+A_{2}+\cdots+A_{k}$ and let $\calB_{k+1}=B_{2}+B_{3}+\cdots+B_{k+1}$.
	Since all the sets are of equal size, we have $|\calA_{k}|=|\calB_{k+1}|$.
	
	The lemma asks us to prove $S-\calA_{\ell}+\calB_{\ell+1}\in\I_{1}$.
	The proof is by induction on $k$. Note that $S-\calA_{1}+\calB_{2}=S-A_{1}+B_{2}\in\I_{1}$
	by the definition of augmenting sets.
	
	Suppose we have already established $X:=S-\calA_{k}+\calB_{k+1}\in\I_{1}$
	and we wish to show $S-\calA_{k+1}+\calB_{k+2}\in\I_{1}$. From definition
	we know that $S-A_{k+1}+B_{k+2}\in\I_{1}$. So, $Y:=S-A_{k+1}+B_{k+2}-\calA_{k}$
	is also in $\I_{1}$. Using this definition, we want to show $Y+\calB_{k+1}\in\I_{1}$.
	
	Now, $|X|=|S|$ while $|Y|=|S|-|\calA_{k}|$. Therefore, using the
	exchange principle of matroids on $X$ and $Y$, we assert that there
	exists a set 
	\[
	R\subseteq A_{k+1}\cup\calB_{k+1},~~~|R|=|\calA_{k}|=|\calB_{k+1}|,~~~~\textrm{s.t.}~~~Z:=Y+R\in\I_{1}
	\]
	If $R=\calB_{k+1}$, we are done. Therefore suppose $R\neq\calB_{k+1}$.
	In particular, $|R\cap\calB_{k+1}|<|\calB_{k+1}|$. Below we show
	this leads to a contradiction, completing the proof of the lemma.
	
	Since $Z\in\I_{1}$ with $|Z|=|S|$ and $S-\calA_{k}\in\I_{1}$, we
	know there exists a subset $Q\subseteq Z\setminus(S-\calA_{k})$ such
	that $|Q|=|\calA_{k}|$ and $S-\calA_{k}+Q\in\I_{1}$. Staring at
	$Z\setminus(S-\calA_{k})$, we get that $Q\subseteq\left(R\cap\calB_{k+1}\right)\cup B_{k+2}$.
	
	Now, we use the fact that there is no edge in $G(S)$ from $a\in\calA_{k}$
	to $B_{k+2}$. Using Fact~\ref{fact:1}, we get that $S-\calA_{k}+b\notin\I_{1}$
	for any $b\in B_{k+2}$. That is, the set $Q$ defined above cannot
	intersect $B_{k+1}$ and thus must be a subset of $R\cap\calB_{k+1}$
	itself. Since $|Q|=|\calA_{k}|=|\calB_{k+1}|$, we get $|R\cap\calB_{k+1}|\geq|\calB_{k+1}|$
	which leads to the promised contradiction. 
\end{proof}
The theorem now follows from the above lemma.
Indeed, the sets $B_{2},B_{3},\ldots,B_{\ell+1}$, and in particular,
$B_{2}+\cdots+B_{\ell+1}$, all lie in $\Span(S)$ in $\M_{1}$. Otherwise,
they would lie in the set $D_{1}$ (by definition of the exchange
graph). Thus, $S+B_{1}\in\I_{1}$, and from the following lemma $S'-B_{1}=S-(A_{1}+A_{2}+\ldots+A_{\ell})+(B_{2}+\cdots+B_{\ell+1})\in\I_{1}$.
The theorem follows from Fact~\ref{fact:2}.
\end{proof}

In the remainder of the section, we want to show an \emph{equivalence}
between the family of augmentation sets in $G(S)$, and a collection
of {\em consecutive} shortest augmenting paths in $G(S)$. At a
high-level, this equivalence allows us to run a single phase of Cunningham
in ``one-shot''. The following lemma is the main structural fact.
\begin{lem}
\label{lem:converseaugset} Let $S$ be a common subset of $\I$
and let $G(S)$ be the corresponding exchange graph. Let $X_{1},X_{2},\ldots,X_{k}$
be {\em disjoint} subsets of $S$, and let $Y_{1},Y_{2},\ldots,Y_{k}$
be {\em disjoint} subsets of $\overline{S}$ satisfying $|X_i| = |Y_i|$ for all $1\le i\le k$.

If (a) there is no edge from a vertex in $X_{i}$ to a vertex in $Y_{j}$
with $j>i$, and (b) $S-X_{1}+Y_{1}-X_{2}+Y_{2}\cdots-X_{k}+Y_{k}\in\I_1$, then 
each set $S-X_{i}+Y_{i}\in\I_1$.

Similarly, if (a) there is no edge from a vertex in $Y_{j}$ to a vertex in $X_{i}$
with $i>j$, and (b) $S-X_{1}+Y_{1}-X_{2}+Y_{2}\cdots-X_{k}+Y_{k}\in\I_2$, then 
each set $S-X_{i}+Y_{i}\in\I_2$.
%
%
%\setenumerate{label=(\alph*),noitemsep}
%\begin{enumerate}
%\item $|X_{i}|=|Y_{i}|$ for all $1\leq i\leq k$. 
%\item There is no edge from a vertex in $X_{i}$ to a vertex in $Y_{j}$
%with $j>i$. 
%\item $S-X_{1}+Y_{1}-X_{2}+Y_{2}\cdots-X_{k}+Y_{k}\in\I$. 
%\end{enumerate}
%Then, each set $S-X_{i}+Y_{i}\in\I$. %	Let $A_{i}\subseteq S,B_{i}\subseteq\bar{S}$
%	be disjoint satisfying (i) $|A_{i}|=|B_{i}|$; (ii) $S-A_{1}+B_{1}-A_{2}+B_{2}...$
%	is independent; (iii) no edge goes from $A_{i}$ to $B_{j}$ for $j>i$.
%	Then $S-A_{i}+B_{i}$ is independent.
\end{lem}

\begin{proof} We prove the statement for $\I_1$; the proof for $\I_2$ is analogous.

The base case of $k=2$ has the main ideas and let us first focus on this.
Let $S':=S-X_{1}+Y_{1}-X_{2}+Y_{2}\in\I_1$. Let $T=S-X_{1}\in\I_1$.
Using the exchange property, we know there exists $Q\subseteq Y_{1}+Y_{2}$
of size $|Q|=|X_{1}|=|Y_{1}|$, such that $S-X_{1}+Q\in\I_1$. Since
there is no edge from any vertex in $X_{1}$ to any vertex in $Y_{2}$,
using Fact~\ref{fact:1} we know that $S-X_{1}+y\notin\I_1$ for any
$y\in Y_{2}$. Thus, $Q\subseteq Y_{1}$. Since $|Q|=|Y_{1}|$, we
get $Q=Y_{1}$. That is, $S-X_{1}+Y_{1}\in\I_1$.

For the second set, we need to work a bit more. We know that $S-X_{1}-X_{2}+Y_{2}\in\I_1$
and $S\in\I_1$. Thus, there exists a set $R\subseteq X_{1}+X_{2}$
such that $|R|=|X_{1}|$ and $S-X_{1}-X_{2}+Y_{2}+R\in\I_1$. If $R=X_{1}$,
we are done. If $R\neq X_{1}$, then since $|R|=|X_{1}|$, there must
exist an $x\in X_{2}\cap R$. That is, the set $T:=S-X_{1}-X_{2}+Y_{2}+x\in\I$.
Note $|T|=|S|-|X_{1}|+1$. Applying the exchange property with the
independent set $S-X_{1}$, we get the existence of $y\in T\setminus(S-X_{1})$
such that $S-X_{1}+y\in\I$. Since $X_{2}\subseteq(S-X_{1})$, we
get that $y\neq x$. In particular, $y\in Y_{2}$. But, as established
before, $S-X_{1}+y\notin\I$ for all $y\in Y_{2}$ (since there is
no edge from $X_{1}$ to $Y_{2}$). 

Now we handle the general case $k>2$ by applying the result for $k=2$ twice. By considering $X_1' = X_1+\ldots+X_i,X_2'=X_{i+1}+...+X_k$ (and set $Y_1',Y_2'$ similarly) we have $S-X_1'+Y_1' \in\I_1$. Finally, by considering $X_1''= X_1'-X_i,X_2''=X_i$ we have $S-X_{i}+Y_{i}\in\I_1$.
\end{proof}
Next, we define the notion of consecutive shortest augmenting paths.
Again, each phase of Cunningham's algorithm finds such an ordered
collection. 
\global\long\def\calP{\mathcal{P}}%
 
\begin{defn}
Let $S$ be a common independent set in $\I_{1}\cap\I_{2}$. Let $G(S)$
be the exchange graph with shortest path $2(\ell+1)$. Let $\calP=(p_{1},p_{2},\ldots,p_{k})$
be an ordered collection of {\em consective shortest augmenting
paths} if (a)$\ $each $p_{i}$ is of length $2(\ell+1)$, and (b)$\ $$p_{i}$
is a valid path in the exchange graph $G(S_{i})$ where $S_{i}$ is
obtained after augmenting along the paths $p_{1},\ldots,p_{i-1}$
in that order. 
\end{defn}

\begin{thm}
\label{thm:sap} Let $S$ be a common independent set in $\I_{1}\cap\I_{2}$.
Given an ordered collection $\calP$ of consecutive shortest augmenting
paths in $G(S)$, we can find an augmenting set $\Pi_{\ell}$ of width
$|\calP|$. 
\end{thm}

\begin{proof}
This follows almost immediately from Lemma~\ref{lem:converseaugset}.
Let $|\calP|=w$. Let the path \\
 $p_{i}:=(s,b_{1}^{(i)},a_{1}^{(i)},b_{2}^{(i)},\cdots,b_{\ell}^{(i)},a_{\ell}^{(i)},b_{\ell+1}^{(i)},t)$
where $b_{k}^{(i)}\in D_{2k-1}$ and $a_{k}^{(i)}\in D_{2k}$.

Define $X_{k}:=\{a_{k}^{(1)},a_{k}^{(2)},\cdots,a_{k}^{(w)}\}$
and $Y_{k}=\{b_{k+1}^{(1)},b_{k+1}^{(2)},\cdots,b_{k+1}^{(w)}\}$. Since
augmenting consecutive shortest paths preserve independence, and since augmentations don't introduce
new shortcuts by the monotonicity lemma (Lemma~\ref{lem:cunning-monot}),
we see that the sets satisfy the conditions of Lemma~\ref{lem:converseaugset} with $\I=\I_1$.

Similary, for  $\I_2$ we consider  $X_{k}:=\{a_{k}^{(1)},a_{k}^{(2)},\cdots,a_{k}^{(w)}\}$
and $Y_{k}=\{b_{k}^{(1)},b_{k}^{(2)},\cdots,b_{k}^{(w)}\}$ instead.
\end{proof}
Next, we show that any augmenting set of width $w$ can be peeled
out into a sequence of consecutive shortest augmenting paths. We start
with a few definitions.
\begin{defn}
Let $\Pi_{\ell}:=(B_{1},A_{1},B_{2},\cdots,B_{\ell},A_{\ell},B_{\ell+1})$
be an augmenting set in $G(S)$. Let $S'=S\oplus\Pi_{\ell}$ be as
defined in Theorem~\ref{thm:aug-means-aug}. We denote the exchange
graph of $G(S')$ as $G(S)|\Pi_{\ell}$. 
\end{defn}

\global\long\def\tPi{\widetilde{\Pi}}%
 
\global\long\def\tA{\widetilde{A}}%
 
\global\long\def\tB{\widetilde{B}}%
 
\begin{defn}
\label{def:subset} Let $S$ be a common independent set and $G(S)$
be the exchange graph with shortest path $2(\ell+1)$. Let $\Pi_{\ell}=(B_{1},A_{1},B_{2},\cdots,B_{\ell},A_{\ell},B_{\ell+1})$
and $\tPi_{\ell}=(\tB_{1},\tA_{1},\tB_{2},\cdots,\tB_{\ell},\tA_{\ell},\tB_{\ell+1})$
be two augmeting sets in $G(S)$. We say $\tPi_{\ell}$ {\em contains}
$\Pi_{\ell}$ if $B_{k}\subseteq\tB_{k}$ and $A_{k}\subseteq\tA_{k}$,
for all $k$. We use the notation $\Pi_{\ell}\subseteq\tPi_{\ell}$
to denote this. 
\end{defn}

\begin{defn}
Let $S$ be a common independent set and $G(S)$ be the exchange graph
with shortest path $2(\ell+1)$. Let $\Pi_{\ell}\subseteq\tPi_{\ell}$
be two augmenting sets, the former contained in the latter. Then we
use $\tPi_{\ell}\setminus\Pi_{\ell}$ to denote the sequence $(\tB_{1}\setminus B_{1},\tA_{1}\setminus A_{1},\cdots,\tB_{\ell+1}\setminus B_{\ell+1})$. 
\end{defn}

%The following theorem shows that if $\Pi_\ell \subseteq \tPi_\ell$, then 

\begin{thm}
\label{thm:augconsistent} Let $\Pi_{\ell}\subseteq\tPi_{\ell}$ be
two augmenting sets in $G(S)$. Then, $\tPi_{\ell}\setminus\Pi_{\ell}$
is an augmenting set in $G(S)|\Pi_{\ell}$. 
\end{thm}

\begin{proof}
Let $P_{k}:=\tA_{k}\setminus A_{k}$ and let $Q_{k}=\tB_{k}\setminus B_{k}$.
Let $S'=S+B_{1}-A_{1}+B_{2}-A_{2}+\cdots+B_{\ell+1}$. Note that this
is $S\oplus\Pi_{\ell}$. For brevity, let $G':=G(S')=G(S)|\Pi_{\ell}$.
Let the layers of $G'$ be denoted as $D'_{1},D'_{2},\ldots$. %In particular, $S'\in \I_1$ (Theorem~\ref{thm:aug-means-aug}). 
 Furthermore, since $\tPi_{\ell}$ is an augmenting set, we get that
\[
S'+Q_{1}-P_{1}+Q_{2}+\cdots+Q_{\ell+1}\in\I_{1}\cap\I_{2}
\]
as well.

\global\long\def\calQ{\mathcal{Q}}%
 Using the above, we can show that each $P_{k}\subseteq D'_{2k}$
and $Q_{k}\subseteq D'_{2k-1}$. For instance, we can show $S'+Q_{1}\in\I_{1}$
(which would put $Q_{1}\subseteq D'_{2k-1}$). To see this, let $\calP:=P_{1}+\cdots+P_{\ell}$
and $\calQ=Q_{2}+\cdots+Q_{\ell+1}$. Note, $|\calP|=|\calQ|$ and
$S'+Q_{1}-\calP+\calQ\in\I_{1}$. Now observe $\calQ\cap D'_{1}=\emptyset$;
this follows from Cunningham's monotonicity lemma (Lemma~\ref{lem:cunning-monot}).
Thus, $\calQ\in\Span_{1}(S')$. Thus, the independence of $S+Q_{1}-\calP+\calQ$
implies\footnote{$|S|+|Q_{1}|=\rk_{1}(S+Q_{1}-\calP+\calQ)\leq\rk_{1}(S+Q_{1}+\calQ)\leq\rk_{1}(S+Q_{1})$}
$S+Q_{1}\in\I_{1}$. For the general case one proceeds inductively.
This proves condition (a) of the augmenting set definition.

Condition (b) holds since $\tPi$ and $\Pi$ were valid augmenting
sets. Condition (c),(d) is proved as above (the base case was (c)).

Finally, condition (e) and (f) follow from Lemma~\ref{lem:converseaugset}.
To see this, we need the observation that there is no edge from $P_{i}$
to $Q_{j}$ for $j>i+1$ and from $Q_{j}$ to $P_{i}$ for $i>j$,
in $G'$. To see this, this was true in $G(S)$ (shortest path property).
And since the levels do not change, and the monotonicity lemma (Lemma~\ref{lem:cunning-monot})
this continues to hold. Now Lemma~\ref{lem:converseaugset} (with
$X_{i}$'s and $Y_{i}$'s properly defined) implies conditions (e)
and (f).

Thus, $\tPi_{\ell}\setminus\Pi_{\ell}$ is a valid augmenting set
in $G'=G(S)|\Pi_{\ell}$. %	
%	Furthermore, there is no edge from $P_i$ to $Q_j$ for $j > i+1$ in $G(S)$. 
%	
%	
%	We have $S':=S+B_{0}-A_{1}+B_{1}-A_{2}+B_{2}...-A_{l}+B_{l}\in\mathcal{I}_{1}$
%	and $S'+Q_{0}'-P_{1}'+Q_{1}'-...-P_{l}'+Q_{l}'\in\mathcal{I}_{1}$.
%	By Lemma \ref{lem:converseaugset}, $S'-P_{k}'+Q_{k}'\in\mathcal{I}_{1}$.
%	Similarly, $S'-P_{k+1}'+Q_{k}'\in\mathcal{I}_{2}$. Thus $\{P_{k}',Q_{k}'\}$
%	are augmenting sets for $G(S')$.
\end{proof}
Now we are armed to prove the following theorem. 
\begin{thm}
\label{thm:Equival-augmenting-sets-paths}\label{thm:equiv} Let $S\in\I_{1}\cap\I_{2}$
and $G(S)$ be the exchange graph. Let $\Pi_{\ell}$ be an augmenting
set of width $w$. Then there exists an ordered collection $\calP$
of $w$ consecutive shortest augmenting paths. 
\end{thm}

\begin{proof}
We first show there is \emph{one} shortest augmenting path hitting
every set of $\Pi_{\ell}$ exactly once. We augment on it and apply
Theorem~\ref{thm:augconsistent}.

Let $\Pi_{\ell}=(B_{1},A_{1},B_{2},\ldots,B_{\ell},A_{\ell},B_{\ell+1})$.
We claim there is a path $p=(b_{1},a_{1},b_{2},\ldots,b_{\ell},a_{\ell},b_{\ell+1})$
such that $b_{k}\in A_{k}$ and $a_{k}\in A_{k}$. Indeed, to prove
this it suffice to show that 

\setenumerate{noitemsep,label=(\roman*)}
\begin{enumerate}
\item for $k=1,\ldots,\ell$, for any $b\in B_{k}$, there is an edge $(b,a)\in G(S)$
to some $a\in A_{k}$ 
\item for $k=1,\ldots,\ell$, for any $a\in A_{k}$, there is an edge $(a,b)\in G(S)$
to some $b\in B_{k+1}$.
\end{enumerate}
For (i), notice that $S-A_{k}+B_{k}\in\I_{2}$ (condition (f) of augmenting
set), and so $S-A_{k}+b\in\I_{2}$. The contrapositive of Fact~\ref{fact:1}
implies there must be an edge from $b$ to some vertex in $A_{k}$.

For (ii) notice that since $S-A_{k}+B_{k+1}\in\I_{1}$ (condition
(e) of augmenting set), for any $a\in A_{k}$, we have $\rk_{1}(S-a+B_{k+1})\geq\rk_{1}(S-A_{k}+B_{k+1})=|S|$
since $|A_{k}|=|B_{k+1}|$. Since $\rk_{1}(S-a)=|S|-1$, by exchange
property of matroids there must exist $b\in B_{k+1}$ such that $S-a+b\in\I_{1}$.
That is, $(a,b)\in G(S)$.

\noindent To complete the proof of the theorem, we observe that augmenting
paths are width-$1$ augmenting sets. We apply Theorem~\ref{thm:augconsistent}
to get $w$ consecutive paths. 
\end{proof}

For our algorithms we will not be interested in finding arbitrary
augmenting sets, but in finding \textit{maximal} augmenting sets.

\begin{defn}[Maximal Augmenting Sets]
Let $S\in\I_{1}\cap\I_{2}$ and let $G(S)$ be the corresponding
exchange graph with shortest augmenting path of length $2(\ell+1)$.
An augmenting set $\Pi_{\ell}$ is called \emph{maximal} if there
exists no other augmenting set $\tPi_{\ell}$ containing $\Pi_{\ell}$. 
\end{defn}

\begin{thm}
Let $S\in\I_{1}\cap\I_{2}$ and let $G(S)$ be the corresponding exchange
graph with shortest augmenting path of length $2(\ell+1)$. Let $\Pi_{\ell}$
be a maximal augmenting set. Then there is no augmenting path of length
$2(\ell+1)$ in $G(S)|\Pi_{\ell}$. %	
%		Let $2l$ be the distance between sources and sinks, and $\{A_{k},B_{k}\}_{k}$
%		form maximal augmenting sets. Then after augmenting $\{A_{k},B_{k}\}_{k}$,
%		there is no augmenting path of length 2l from sources to sinks.
\end{thm}

\begin{proof}
This is a corollary of Theorem~\ref{thm:sap} and Theorem~\ref{thm:equiv}.
Let $\calP$ be the collection of paths obtained by invoking Theorem~\ref{thm:equiv}
on $\Pi_{\ell}$. If there is another shortest augmenting path $p\in G(S)|\Pi_{\ell}$,
then $(\calP,p)$ is an ordered collection of consecutive shortest
augmenting paths in $G(S)$. Thus Theorem~\ref{thm:sap} invoked
on this returns $\tPi_{\ell}$ which contains $\Pi_{\ell}$. 
\end{proof}
The above lemma shows that a single phase of Cunningham's algorithm
corresponds to finding maximal augmenting sets. This is what we do
in Section~\ref{subsec:TheAugSetsAlgo}, but we need
one more concept before: \textit{partial augmenting sets}.

\subsubsection{Partial Augmenting Sets\label{subsec:Partial-Aug-Sets-Approach}}

%In the last section we showed that augmenting paths preserve independence.
%Therefore, a natural approach to solving matroid intersection is to
%find a \emph{maximal} collection of augmenting sets. Once such a maximal
%collection is augmented, the distance from sources to sinks increases
%by at least 2 and we can move on to the next phase. Recall that a
%$(1-\epsilon)$-approximate solution requires only $O(1/\epsilon)$
%phases.
In order to find maximal augmenting sets, our algorithm will leverage objects which are not augmenting sets at all. These are what we call
\emph{partial augmenting sets}, and the final algorithm will work work by
incrementally finding better partial augmenting sets\textit{\emph{,
from which augmenting sets can be extracted.}} 
\begin{defn}[Partial Augmenting Sets]
Let $S\in\I_{1}\cap\I_{2}$ and $G(S)$ be the corresponding exchange
graph with shortest path $2(\ell+1)$. A collection of sets $\Phi_{\ell}:=(B_{1},A_{1},B_{2},A_{2},\ldots,A_{\ell},B_{\ell+1})$
%Sets $\{A_{k},B_{k}\}_{k}$ 
 form a \emph{partial augmenting set} if the following conditions
are satisfied. 
\global\long\def\width{w}%

\setenumerate{noitemsep,label=(\alph*)}
\begin{enumerate}
\item For $1\leq k\leq\ell+1$, we have $A_{k}\subseteq D_{2k}$ and $B_{k}\subseteq D_{2k-1}$.
\item $|B_{1}|\ge|A_{1}|\ge|B_{2}|\ge\cdots\ge|B_{\ell+1}|$. 
\item $S+B_{1}\in\mathcal{I}_{1}$ 
\item $S+B_{\ell+1}\in\mathcal{I}_{2}$ 
\item For all $1\leq k\leq\ell$, we have $S-A_{k}+B_{k+1}\in\mathcal{I}_{1}$. 
\item For all $1\leq k\leq\ell$, we have $\rk_{2}(S-A_{k}+B_{k})=\rk_{2}(S)$.
\end{enumerate}
%	We let $\width$ denote the {\em width} of $\Pi_\ell$.
\end{defn}

%\begin{defn}[Partial Augmenting Sets]
%Sets $\{A_{k},B_{k}\}_{k}$ forms \emph{partial} augmenting sets
%if $B_{0}\subseteq D_{0},A_{1}\subseteq D_{1},B_{1}\subseteq D_{2},...,B_{l}\subseteq D_{2l}$
%satisfy the following:
%\begin{itemize}
%\item $|B_{0}|\geq|A_{1}|\geq|B_{1}|\geq...\geq|A_{l}|\geq|B_{l}|$
%\item $S+B_{0}\in\mathcal{I}_{1}$
%\item $S-A_{k}+B_{k}\in\mathcal{I}_{1}$
%\item $S+B_{l}\in\mathcal{I}_{2}$
%\item $\rk_{2}(S-A_{k}+B_{k-1})=\rk_{2}(S)$
%\end{itemize}
%\end{defn}
As in Definition~\ref{def:subset}, we say $\Pi_{\ell}\subseteq\Phi_{\ell}$
if each ``coordinate'' of $\Pi_{\ell}$ is a subset of the corresponding
``coordinate'' of $\Phi_{\ell}$. The following lemma gives an efficient
algorithm to convert any partial augmenting set $\Phi_{\ell}$ into
a true augmenting set $\Pi_{\ell}$. For now, the reader should consider
$\Pi'_{\ell}$ to be a string of empty sets; we will invoke this lemma
later with non-empty $\Pi'_{\ell}$. %For now readers can think of
%$\{A_{k}',B_{k}'\}_{k}$ as empty sets. Later we will apply this lemma
%to nonempty $\{A_{k}',B_{k}'\}_{k}$ in the analysis of our algorithm.

\begin{lem}
\label{lem:partial-to-full-aug-sets} Given a partial augmenting set
$\Phi_{\ell}=(B_{1},A_{1},B_{2},\cdots,B_{\ell},A_{\ell},B_{\ell+1})$
that contains an augmenting set $\Pi'_{\ell}=(B'_{1},A'_{1},\ldots,B'_{\ell})$,
in $O(n\cdot\Tind)$-time we can find an augmenting set $\tPi_{\ell}=(\tB_{1},\tA_{1},\tB_{2},\cdots,\tB_{\ell},\tA_{\ell},\tB_{\ell+1})$
such that

\setenumerate{noitemsep,label=(\alph*)}
\begin{enumerate}
\item $\Pi'_{\ell}\subseteq\tPi_{\ell}\subseteq\Phi_{\ell}$, and 
\item $\tB_{\ell+1}=B_{\ell+1}$.
\end{enumerate}
%$\{A_{k},B_{k}\}_{k}$ that contain augmenting sets $\{A_{k}',B_{k}'\}_{k}$,
%we can find in $O(n)$ time some augmenting sets $\{A_{k}'',B_{k}''\}_{k}$
%that satisfy $B_{l}''=B_{l}$ and $A_{k}'\subseteq A_{k}''\subseteq A_{k},B_{k}'\subseteq B_{k}''\subseteq B_{k}$.
\end{lem}

\begin{proof}
We work backwards from $\tB_{\ell+1}=B_{\ell+1}$. Let $w=|B_{\ell+1}|$.
Suppose we have constructed sets $\tB_{\ell+1},\tA_{\ell},\cdots,\tB_{k+1}$.
Now, since $S-A_{k}+\tB_{k+1}\subseteq S-A_{k}+B_{k+1}$, and since
the latter lies in $\I_{1}$ (property (e) of partial augmenting sets),
we have $S-A_{k}+\tB_{k+1}\in\I_{1}$. We also have $S-A_{k}'+B_{k}'\in\mathcal{I}_{1}$.
Therefore, by the exchange property, we can find a set $\tA_{k}$
of size $|\tA_{k}|=|\tB_{k+1}|=w$ such that $A'_{k}\subseteq\tA_{k}\subseteq A_{k}$
and $S-\tA_{k}+\tB_{k+1}\in\I_{1}$. The time taken to do so is $O(|A_{k}|)$
independence-query oracles. %we can select $A_{k}''$ of size $|B_{k}''|=|B_{l}|$ such that $A_{k}'\subseteq A_{k}''\subseteq A_{k}$
%and $S-A_{k}''+B_{k}''\in\mathcal{I}_{1}$.

Now, suppose we have constructed $\tB_{\ell+1},\tA_{\ell},\cdots,\tA_{k}$
with the desired property. Then since $\tA_{k}\subseteq A_{k}$ and
$\rk_{2}(S-A_{k}+B_{k})=\rk_{2}(S)$, we get $\rk_{2}(S-\tA_{k}+B_{k})\geq\rk_{2}(S)$.
Also, since $S-A'_{k}+B'_{k}\in\I_{2}$, we get $S-\tA_{k}+B'_{k}\in\I_{2}$
as well. Thus, we can select a subset $\tB_{k}$ such that $B'_{k}\subseteq\tB_{k}\subseteq B_{k}$
and $S-\tA_{k}+\tB_{k}\in\I_{2}$ and has size $|S|$. That is, $|\tB_{k}|=|\tA_{k}|=w$.
The time taken to find this is $O(|B_{k}|)$ independence oracle queries.
%This implies, together with the independence of $S-A_{k+1}''+B_{k}'\subseteq S-A_{k+1}'+B_{k}'$
%that we can select some $B_{k}''$ s.t. $B_{k}'\subseteq B_{k}''\subseteq B_{k}$
%and $S-A_{k+1}''+B_{k}''\in\mathcal{I}_{2}$ and has size $|S|$,
%i.e. $|B_{k}''|=|A_{k+1}''|=|B_{l}|$.
%The running time is $O(n)$ because we only need to attempt to add
%each element at most once.

The total running time is $\sum_{k}O(|A_{k}|+|B_{k}|)$ many independence
oracle calls which is $O(n\cdot\Tind)$-time. 
\end{proof}
%Our algorithm will continue finding better \textit{partial augmenting
%sets} until they reach a ``maximal'' state. This motivates us to
%define \textit{maximal augmenting sets.}
%\begin{defn}[Maximal Augmenting Sets]
%Augmenting sets $\{A_{k},B_{k}\}_{k}$ are called \textit{maximal}
%if there does not exist any augmenting set of larger size containing
%them.
%\end{defn}
%As a corollary of Theorem \ref{thm:equiv}, we get that the current
%phase ends if we perform augmentations along some maximal augmenting
%sets.
%\begin{cor}
%Let $2l$ be the distance between sources and sinks, and $\{A_{k},B_{k}\}_{k}$
%form maximal augmenting sets. Then after augmenting $\{A_{k},B_{k}\}_{k}$,
%there is no augmenting path of length 2l from sources to sinks.
%\end{cor}

Now we are ready to prove a {\em key} property which relates the
width of any two maximal augmenting sets. This will be crucial to
give the subquadratic approximation algorithm in Section$\ $\ref{subsec:FasterAlgo}.
%Next we prove a property that any two collections of maximal augmenting
%sets have comparable sizes. This will be useful for our hybrid algorithm
%in Section \ref{subsec:FasterAlgo}.
 
\global\long\def\tw{\widetilde{w}}%
 
\begin{lem}
\label{lem:maximal-aug-sets-size-ratio} Let $S\in\I_{1}\cap\I_{2}$
and $G(S)$ be the corresponding graph with shortest augmenting path of
length $2(\ell+1)$. Let $\Pi_{\ell}$ and $\tPi_{\ell}$ be two maximal
augmenting sets of width $w$ and $\tw$. Then, $\tw\leq(2\ell+4)w$. 
\end{lem}

\begin{proof}
Let $\Pi_{\ell}=(B_{1},A_{1},B_{2},\ldots,A_{\ell},B_{\ell+1})$.
Let $\tPi_{\ell}=(Q_{1},P_{1},Q_{2},\ldots,P_{\ell},Q_{\ell+1})$.
Note $|A_{k}|=|B_{j}|=w$ and $|P_{k}|=|Q_{j}|=\tw$. For the sake
of contradiction, assume $\tw>(2\ell+4)w$. %	
%Let $\{A_{k},B_{k}\}_{k}$ and $\{P_{k},Q_{k}\}_{k}$ be maximal  augmenting
%sets. Assume the contrary that $|A_{k}|<(2l+2)\cdot|P_{k}|$. 

We construct partial augmenting sets $\Phi_{\ell}=(B'_{1},A'_{1},\cdots,A'_{\ell},B'_{\ell+1})$
such that (i) $A_{k}\subseteq A_{k}'\subseteq A_{k}+P_{k}$, (ii)
$B_{k}\subseteq B_{k}'\subseteq B_{k}+Q_{k}$, (iii) $|A_{k}'|>(2\ell+3-2k)\cdot|A_{k}|$,
and (iv) $|B_{k}'|>(2\ell+4-2k)\cdot|B_{k}|$, and (v) $|B'_{\ell+1}|>|B_{\ell+1}|$.
Note that the last point would imply %Note that this would imply $|B'_{\ell+1}| > |B_{\ell+1}|$, that is, 
$B_{\ell+1}$ is a \emph{strict} subset of $B'_{\ell+1}$. From Lemma~\ref{lem:partial-to-full-aug-sets},
we would get an augmenting set $\Pi'_{\ell}$ which would contain $\Pi_{\ell}$,
contradicting the maximality of the latter. \medskip{}

%This would imply $|B_{l}'|>|B_{l}|$. From the last lemma $\{A_{k},B_{k}\}_{k}$
%is then contained in some larger augmenting sets which contradicts
%the maximality of $\{A_{k},B_{k}\}_{k}$.
Now we describe the construction of $\Phi_{\ell}$.

\noindent To construct $B'_{1}$, we start with $B_{1}$ and keep
adding elements of $Q_{1}$ to it maintaining $S+B'_{1}\in\I_{1}$.
By the matroid exchange property, we will add $|Q_{1}\setminus B_{1}|$
such elements giving $|B'_{1}|=|Q_{1}|=\tw>(2\ell+4)w$. We satisfy
(i), (iv), and property (c) of the partial augmenting set.

Suppose we have constructed $B_{1}',A_{1}',B_{2}',...,B_{k}'$ as
desired. We need to construct $A'_{k}$. Since $S-A_{k}+B_{k}\in\mathcal{I}_{2}$,
we get $S-\left(P_{k}+A_{k}\right)+B_{k}\in\mathcal{I}_{2}$. We also
have $S-P_{k}+Q_{k}\in\I_{2}$. As $B'_{k}\subseteq B_{k}+Q_{k}$,
we get that $S-P_{k}+(B_{k}'\backslash B_{k})\in\mathcal{I}_{2}$.
By the exchange property, there must exist a $Z\subseteq B'_{k}\setminus B_{k}$
such that $S-\left(P_{k}+A_{k}\right)+B_{k}+Z\in\I_{2}$ and % there
%must some $Q\subseteq B_{k}'\backslash B_{k}$ for which $S-P_{k+1}-A_{k+1}+B_{k}+Q\in\mathcal{I}_{2}$
%and 
\[
|B_{k}+Z|=|B_{k}'\backslash B_{k}|\geq|B_{k}'|-|B_{k}|>(2\ell+4-2k)|B_{k}|-|B_{k}|=(2\ell+3-2k)w
\]
where the inequality follows from (iv) above (which we assume we have
maintained so far). Now since $S-A_{k}+B_{k}\in\mathcal{I}_{2}$,
there is some $A_{k}'$ of size $|B_{k}+Z|$ s.t. $A_{k}\subseteq A_{k}'\subseteq A_{k}+P_{k+1}$
and $\rk_{2}(S-A_{k}'+B_{k}+Z)=\rk_{2}(S)$. Thus, property (d) of
the partial augmenting set is satisfied, and so is property (iii).

Now, suppose we have constructed $B_{1}',A_{1}',B_{2}',...,A_{k}'$
as desired. We need to construct $B'_{k+1}$. As $S-P_{k}+Q_{k+1}\in\mathcal{I}_{1}$,
there must exist some $Z\subseteq Q_{k}$ for which $S-A_{k}'+Z\in\mathcal{I}_{1}$
and 
\[
|Z|=|A_{k}'\cap P_{k}|\geq|A_{k}'|-|A_{k}|>(2\ell+3-2k)|A_{k}|-|A_{k}|=(2\ell+2-2k)w.
\]

Now consider $S-A_{k}'+B_{k+1}$ which is independent since $A_{k}\subseteq A_{k}'$.
We can then find $B_{k+1}'$ of size $|Z|$ s.t. $B_{k+1}\subseteq B_{k+1}'\subseteq B_{k+1}+Z$
and $S-A_{k}'+B_{k+1}'\in\mathcal{I}_{1}$. Note that property (iv)
is satisfied for $|B'_{k+1}|$. And the above satisfies property (e)
of the partial augmenting sets.

Finally, we need to satisfy condition (c), that is, $S+B'_{\ell+1}\in\I_{2}$.
After the above procedure, we do get $B_{\ell+1}\subseteq B'_{\ell+1}\subseteq B_{\ell+1}+Q_{\ell+1}$
satisfying every property except perhaps this one. Also, $|B'_{\ell+1}|>(2\ell+4-2(\ell+1))w=2w$.
Now, since $S+Q_{\ell+1}\in\I_{2}$, we know that $B'_{\ell+1}\setminus B_{\ell+1}\in\I_{2}$.
Thus, we keep deleting items from $B'_{\ell+1}$ such that we get
$S+B'_{\ell+1}\in\I_{2}$. We will delete at most $|B_{\ell+1}|=w$
elements. The final $B'_{\ell+1}$ will satisfy condition (c) and
$|B'_{\ell+1}|>w$, that is, property (v). This completes the proof. 
\end{proof}

\subsection{The Augmenting Sets Algorithm\label{subsec:TheAugSetsAlgo}}

\subsubsection{An Analogy and High level idea}

\label{subsec:analogy}

Before presenting our augmenting sets algorithm, we provide an analogy
to finding a maximal collection of disjoint shortest augmenting paths.
Recall that the Hopcroft-Karp idea is inapplicable because augmenting
such a collection do not necessarily preserve independence. However,
let us pretend for the moment that this was indeed our task. Our algorithm
in the next section is in fact inspired by this.

Suppose that we have constructed internally disjoint partial augmenting
paths from the source. Our goal is to extend them all the way to the
sink. Let $A_{k}\subseteq D_{2k}$ be the sets of vertices on the
current augmenting sets in $D_{2k}$. Define $B_{k}$ similarly. We
then have $|B_{1}|\geq|A_{1}|\geq...\geq|B_{\ell+1}|$.

Let us focus on $A_{k}$ and $B_{k+1}$. Because they represent the
vertices on partial augmenting paths, currently we have a \emph{matching}
between a subset of $A_{k}$ and $B_{k+1}$. A natural idea is then
to match more vertices of $D_{2k+1}$ with the unmatched vertices
of $A_{k}$. Let $B\subseteq D_{2k+1}$ be such a maximal subset,
and $A\subseteq A_{k}$ be the vertices that are still not matched.
Such a maximal set requires only $O(|D_{2k+1}|)$ independence oracle
calls.

The crucial observation is that all of $A$'s outgoing edges go to
$B_{k}+B$ as $B$ is maximal. However, $B_{k}+B$ is already fully
matched; so $A$ is essentially useless as its vertices would never
get matched and can be removed from future consideration. It is easy
to verify that a total of $|B_{k}|-|A_{k}|$ elements are matched
or become useless.

A similar operation exists for extending the matching between $B_{k}$
and $A_{k+1}$. By repeating these operations we will eventually find
a maximal collection of disjoint shortest augmenting paths. Our treatment
of this similar scenario forms the basis of the augmenting sets algorithm
given below.

\subsubsection{The Detailed Algorithm}

In this section we describe a (slow) $O(n^{2}\cdot\indeptime)$-time
algorithm to find a maximal augmenting set. It is an iterative algorithm
which finds ``better'' partial augmenting sets in each iteration.
The algorithm assumes access to the layers $D_{k}$'s of the current
exchange graph $G(S)$ which has shortest path length $2(\ell+1)$.

%We first describe a slow $O(n^{2})$ time augmenting sets algorithm
%before combining it with our previous augmenting path algorithm to
%obtain our desired runtime of $O(n^{1.5}/\epsilon^{1.5})$. Our augmenting
%sets algorithm incrementally finds better partial augmenting sets.

We maintain three types of elements in each layer : \textbf{selected}
elements denoted as $A_{k},B_{k}$, \textbf{removed} elements denoted
as $R_{k}$, and \textbf{fresh} $F_{k}$. The \textit{selected} elements
$A_{k}\subseteq D_{2k},B_{k}\subseteq D_{2k-1}$ form a partial augmenting
set, and our algorithm incrementally finds better and better partial
augmenting sets with larger $B_{\ell+1}$. The \textit{removed} elements
$R_{k}\subseteq D_{k}$ are deemed to be useless for finding better
partial augmenting sets. Finally, \textit{fresh} elements $F_{k}\subseteq D_{k}$
are those that are neither selected nor removed. The type of an element
can change from \emph{fresh to selected} and from \emph{selected to
removed} but never the other way.

Initially, $B_{1}$ is a maximal subset of $\overline{S}$ that can
be added to $S$ while being independent in $\I_{1}$. Subsequently,
$F_{1}=D_{1}\backslash B_{1}$ and $R_{1}=\emptyset$. For $2\leq k\leq2\ell+1$,
we initialize sets $A_{k},B_{k},R_{k}=\emptyset$ and $F_{k}=D_{k}$.
For convenience we set $A_{0}=R_{0}=F_{0}=D_{0}=\emptyset$ and $A_{\ell+1}=R_{2\ell+2}=F_{2\ell+2}=D_{2\ell+2}=\emptyset$.

We maintain the following \emph{invariants}. It is easy to verify
that the initial conditions satisfy them. 

\setenumerate{noitemsep,label=(\alph*)}
\begin{enumerate}
\item For $1\leq k\leq\ell$, we have~~$S-A_{k}+B_{k+1}\in\mathcal{I}_{1}$.
\item $\rk_{2}(S-A_{k}+B_{k})=\rk_{2}(S)$. \\
Equivalently, $\exists B\subseteq B_{k}$ of size $|B|=|A_{k}|$ for
which $S-A_{k}+B\in\mathcal{I}_{2}$. 
\item For $1\leq k\leq\ell$, for any $X\subseteq B_{k+1}+F_{2k+1}=D_{2k+1}-R_{2k+1}$,
if $S-(A_{k}+R_{2k})+X\in\mathcal{I}_{1}$ then $S-A_{k}+X\in\mathcal{I}_{1}$. 
\item $R_{2k-1}\in\Span_{2}(S-(D_{2k}-R_{2k})+B_{k})$
\end{enumerate}
Invariants (a) and (b) ensure that the selected sets $\Phi_{\ell}:=(B_{1},A_{1},B_{2},\cdots,A_{\ell},B_{\ell+1})$
form a partial augmenting set. Invariant (c) essentially says that
if $R_{2k+1}$ is ``useless'' (note $X\cap R_{2k+1}$ is empty),
so is $R_{2k}$. On the other hand, Invariant (d) says that if $R_{2k}$
is ``useless'', so is $R_{2k-1}$. Invariants (c) and (d) are important
for certifying that removed elements do not need be considered at
all (see Lemma \ref{lem:maximal-aug-sets}). %Without them an element can
%be repeatedly added and removed from our candidate partial augmenting
%sets, and the algorithm would not even terminate.

We are ready to present $\texttt{Refine}$, which consists of a series
of operations on the pairs $(A_{k},B_{k+1})$ and $(B_{k},A_{k})$
inspired by the matching analogy given in Section~\ref{subsec:analogy}.
For $(A_{k},B_{k+1})$, we extend $B_{k+1}$ as much as possible while
respecting Invariant (a) (Lines 1-2 of $\texttt{Refine1}(k)$). We
then identify the ``matched'' vertices in $A_{k}$ and remove the
ones still not matched (Lines 3-4 of $\texttt{Refine1}(k)$).

Similarly, for $(B_{k},A_{k})$, we find a maximal subset of $B_{k}$
that can be ``matched'' and remove the rest (lines 1-2 of $\texttt{Refine2}(k)$).
We then find the ``endpoints'' in $A_{k}$ that are ``matched''
to $B_{k}$ (lines 3-4 of $\texttt{Refine2}(k)$).

\begin{algorithm2e}[H]

\caption{$\texttt{Refine1}(k)$}

\SetAlgoLined

Find maximal $B\subseteq F_{2k-1}$ s.t. $S-A_{k}+B_{k+1}+B\in\mathcal{I}_{1}$

$B_{k+1}\longleftarrow B_{k+1}+B,F_{2k-1}\longleftarrow F_{2k-1}-B$

Find maximal $A\subseteq A_{k}$ s.t. $S-A_{k}+B_{k+1}+A\in\mathcal{I}_{1}$

$A_{k}\longleftarrow A_{k}-A,R_{2k}\longleftarrow R_{2k}+A$

\end{algorithm2e}

~

\begin{algorithm2e}[H]

\caption{$\texttt{Refine2}(k)$}

\SetAlgoLined

Find maximal $B\subseteq B_{k}$ s.t. $S-(D_{2k}-R_{2k})+B\in\mathcal{I}_{2}$

$R_{2k-1}\longleftarrow R_{2k-1}+B_{k}\backslash B,B_{k}\longleftarrow B$

Find maximal $A\subseteq F_{2k}$ s.t. $S-\left(D_{2k}-R_{2k}\right)+B_{k}+A\in\mathcal{I}_{2}$

$A_{k}\longleftarrow A_{k}+F_{2k}\backslash A,F_{2k}\longleftarrow A$

\end{algorithm2e}

~

\begin{algorithm2e}[H]

\caption{$\texttt{Refine}$}

\SetAlgoLined

\For{$k=0,1,\cdots,l$}{

$\texttt{Refine1}$($k$)

$\texttt{Refine2}$($k$)

}

$\texttt{Refine1}$($0$)

\end{algorithm2e}

It is somewhat strange that $\texttt{Refine}$ ends with $\texttt{Refine1}$(0)
which is indeed redundant. However it makes the analysis of the combined
algorithm in the next subsection less cumbersome (specifically Lemma
\ref{lem:boundingpartial}).

The next three lemmas study properties of $\texttt{Refine1}$ and
$\texttt{Refine2}$ which will be crucial for analyzing the performance
of our main subroutine $\texttt{Refine}$.
\begin{lem}
\label{lem:equal-obs} After $\texttt{Refine1}(k)$ is run, we have
$|A_{k}|=|B_{k+1}|$. After $\texttt{Refine2}(k)$ is run, we have
$|B_{k}|=|A_{k}|$. 
\end{lem}

\begin{proof}
We prove the first statement; the second statement's proof is analogous.
Let $A_{k}^{old}$ and $B_{k+1}^{old}$ be the sets before $\texttt{Refine1}(k)$
is called. In Line 2, we set $B_{k+1}$. At this point, let $T:=S-A_{k}^{old}+B_{k+1}$;
we know $T\in\I_{2}$. We also know $\rk_{1}(T+A_{k}^{old})=\rk_{1}(S+B_{k+1})=|S|$
since $B_{k+1}$ is in the span of $S$. This implies, the maximal
set $A$ added in Line 4 leads to $A_{k}$ such that $\rk_{1}(S-A_{k}+B_{k+1})=|S|$,
implying $|A_{k}|=|B_{k+1}|$. 
\end{proof}
\begin{lem}[Properties of $\texttt{Refine1}$]
Let $A_{k}^{old},B_{k+1}^{old}$ be the sets ($A_{k},B_{k+1}$) before
$\texttt{Refine1}(k)$ is called. After the call, all four invariants
are preserved. Moreover, for $k\geq1$ a total of $|A_{k}^{old}|-|B_{k+1}^{old}|$
elements which were formerly fresh are now selected, or formerly selected
but now removed. 
\end{lem}

\begin{proof}
After running $\texttt{Refine1}(k)$, we have $|A_{k}|=|B_{k+1}|$
by Lemma~\ref{lem:equal-obs}. Also note that $|B_{k+1}|-|B_{k+1}^{old}|$
elements are selected and $|A_{k}^{old}|-|A_{k}|$ are removed. Thus
a total of $|A_{k}^{old}|-|B_{k}^{old}|$ elements change status.

Invariant (a) is true by design.

Invariant (b) remains true because $B_{k}$ is unchanged and $A_{k}$
only loses elements.

Invariant (d) remains true because $R_{2k}$ is only increased, and
thus its span can only become bigger. %unchanged. Since $S-A_{k}^{old}+B_{k}^{old}\in\mathcal{I}_{1}$, Invariant
%1 remains true by design.

For Invariant (c), first notice $A_{k}+R_{2k}=A_{k}^{old}+R_{2k}^{old}$.
Since the invariant held before run, we know for any $X\subseteq D_{2k+1}-R_{2k+1}$,
if $S-(A_{k}^{old}+R_{2k}^{old})+X\in\mathcal{I}_{1}$, then $S-A_{k}^{old}+X\in\I_{1}$.
That is, $S-(A_{k}+R_{2k})+R_{2k}^{old}+X\in\mathcal{I}_{1}$. Our
goal is to show that $S-A_{k}^{old}+A+X\in\mathcal{I}_{1}$.

Since $S-A_{k}^{old}+A+B_{k+1}\in\mathcal{I}_{1}$, we can find $\bar{A}\subseteq A$
and $\bar{B}\subseteq B_{k+1}\backslash X$ s.t. $S-A_{k}^{old}+\bar{A}+X+\bar{B}\in\I_{1}$
and $|\bar{A}|+|\bar{B}|=|A|+|B_{k+1}|-|X|$. Now if $\bar{A}\neq A$
then $|\bar{B}|>|B_{k+1}|-|X|$ and hence $X+\bar{B}\subseteq D_{2k+1}-R_{2k+1}$
has size bigger than $|B_{k+1}|$. But this contradicts the fact that
$B_{k+1}\subseteq D_{2k+1}-R_{2k+1}$ is maximally independent in
$S-A_{k}^{old}$. Thus we must have $\bar{A}=A$ and $S-A_{k}^{old}+A+X\in\mathcal{I}_{1}$. 
\end{proof}
\begin{lem}[Properties of $\texttt{Refine2}$]
 Let $A_{k}^{old},B_{k}^{old}$ be the $(A_{k},B_{k})$ before call
to $\texttt{Refine2}(k)$. After the call, all four invariants are
preserved. Moreover, for $k\leq\ell$ a total of $|B_{k}^{old}|-|A_{k}^{old}|$
elements which were formerly fresh but are now selected, or were formerly selected
but now removed. 
\end{lem}

\begin{proof}
After running $\texttt{Refine2}(k)$, we have $|A_{k}|=|B_{k}|$ by
Lemma~\ref{lem:equal-obs}. %It is easy to verify the lemma for $B_{2l}$. For $k<2l$, 
Also note that $|B_{k}^{old}|-|B_{k}|$ elements removed and $|A_{k}|-|A_{k}^{old}|$
selected for a total of $|B_{k}^{old}|-|A_{k}^{old}|$ elements changing
status.

Invariants (a) and (c) remain true because $A_{k-1}$ is unchanged
and $B_{k}$ only loses elements. For Invariant (b), observe that
$\texttt{Refine2}$ simply finds a maximally independent subset of
$S-(D_{2k}-R_{2k})+B_{k}^{old}+F_{2k}^{old}=S-A_{k}^{old}+B_{k}^{old}$
so $\rk_{2}(S-A_{k}+B_{k})$ remains unchanged.

Invariant (d) remains true essentially by design. $\texttt{Refine2}$
finds a base of $S-(D_{2k}-R_{2k})+B_{k}^{old}$ and moves any of
element not in the base into $R_{2k-1}$. Those elements are in the
span by definition. 
\end{proof}
Now by summing over all $k$ in the last two lemmas, we get the following
corollary. 
\begin{cor}
\label{cor:After-Refine} After applying $\texttt{Refine}$, at least
a total of $|B_{1}^{old}|-|B_{\ell+1}^{old}|$ elements are formerly
fresh but now selected or formerly selected but now removed. 
\end{cor}

\begin{proof}
This essentially follows from the last two lemmas. Readers may notice
that the sum does not exactly telescope because at the time $\texttt{Refine2}(k)$
is applied, $B_{k}$ has been modified by $\texttt{Refine1}(k)$.
However this is okay as $\texttt{Refine1}(k)$ only adds elements
to the old $B_{k}$. Same for $\texttt{Refine2}(k)$ which only adds
elements to $A_{k+1}$ and can only help $\texttt{Refine1}(k+1)$. 
\end{proof}
Using the properties established above, we show how $\texttt{Refine}$
can be repeatedly used to find maximal augmenting sets. 
\begin{lem}
\label{lem:maximal-aug-sets}Suppose we run $\texttt{Refine}$ until
no more element changes its type. At this point, the collection $\Pi_{\ell}=(B_{1},A_{1},B_{2},A_{2},\cdots,A_{\ell},B_{\ell+1})$
is a maximal augmenting set.
\end{lem}

\begin{proof}
For the sake of contradiction, assume there exists $\tPi_{\ell}$
containing $\Pi_{\ell}$. Then by Theorem~\ref{thm:augconsistent}
and Theorem~\ref{thm:Equival-augmenting-sets-paths}, there is a
sequence %Lemma \ref{lem:extendaugset}
%and Corollary \ref{cor:Equival-augmenting-sets-paths}, its difference
%with $\{A_{k},B_{k}\}_{k}$ consists of a sequence 
of shortest augmenting paths in $G(S)|\Pi_{\ell}$. Taking the first
path in the sequence, we have a shortest path $(b_{1},a_{1},b_{2},a_{2},\ldots,a_{\ell},b_{\ell+1})$
in $G(S)|\Pi_{\ell}$. In particular, $a_{k}\notin A_{k}$ and $b_{k}\notin B_{k}$
for any $k$. % from sources
%to sinks such that $\{A_{k}+a_{k},B_{k}+b_{k}\}_{k}$ forms augmenting
%sets ($a_{k}\notin A_{k},b_{k}\notin B_{k}$).

We claim that all $a_{k},b_{k}$ must have been removed by $\texttt{Refine}$.
This would be a contradiction because an element in $D_{2\ell+1}$
is removed only if it is in $\Span_{2}(S+B_{\ell+1})$. But in that
case $b_{\ell+1}$ would not be in $D'_{2\ell+1}$ of $G'=G(S)|\Pi_{\ell}$.
We prove this claim by induction.

First, notice $b_{1}$ cannot be fresh. Otherwise $b_{1}$ would have
been selected by $\texttt{Refine}$ since $S+B_{1}+b\in\mathcal{I}_{1}$.
So $b_{1}$ must have been removed by $\texttt{Refine}$. Now we have
two cases.

Suppose $b_{1},a_{1},...,b_{k}$ have been removed. We need to show
$a_{k}$ is removed. Suppose not, and suppose $a_{k}\in F_{2k}$.
Since $S-(A_{k}+a_{k})+B_{k}+b_{k}\in\mathcal{I}_{2}$, $S-(D_{2k}-R_{2k})+B_{k}+b_{k}\in\mathcal{I}_{2}$.
But this is a contradiction as $b_{k}\in R_{2k-1}\subseteq \Span_2(S-D_{2k}+R_{2k}+B_{k})$.

Now, suppose $b_{1},a_{1},...,a_{k}$ have been removed. We need to
show $b_{k+1}$ is removed. Suppose not, and suppose $b_{k+1}\in F_{2k+1}$.
Then $X=B_{k+1}+b_{k+1}\subseteq B_{k+1}+F_{2k+1}$. Since $S-A_{k}-a_{k}+X\in\mathcal{I}_{1}$
and $a_{k}\in R_{2k}$, $S-A_{k}-R_{2k}+X\in\mathcal{I}_{1}$. Invariant
(c) then implies $S-A_{k}+X\in\mathcal{I}_{1}$, which is a contradiction
since $|X|=|A_{k}|+1$ and $X\in\Span_{1}(S)$ (because $k+1>1$).
\end{proof}
\begin{lem}
\label{lem:Refine-RunTime} $\texttt{Refine1}(k)$ makes $O(|D_{2k}|+|D_{2k+1}|)$
independence oracle calls. $\texttt{Refine2}(k)$ makes $O(|D_{2k}|+|D_{2k-1}|)$
independence oracle calls. $\texttt{Refine}$ takes $O(n\cdot\indeptime)$-time.
\end{lem}

\begin{proof}
Adding maximal $L\subseteq M$ to an independent set $N$ s.t. $N+L$
is still independent requires only adding elements of $L$ one at
a time. This establishes the runtime of $\texttt{Refine1}$ and $\texttt{Refine2}$.
Now $\texttt{Refine}$ takes $\sum_{k}O(|D_{2k}|+|D_{2k+1}|)+O(|D_{2k}|+|D_{2k-1}|)=O(n)$
time since $n\geq\sum_{k}|D_{k}|$. 
\end{proof}
At this point, we can get an $O(n^{2}/\epsilon \cdot\indeptime)$-time algorithm.
The algorithm runs in phases; in each phase the shortest augmenting
path length of the exchange graph goes up by $\geq2$. In each phase,
we use $\texttt{Refine}$ till we find the maximal augmenting set
as given by Lemma~\ref{subsec:TheAugSetsAlgo}. Since each element
can change its type at most twice, we know that we need only $O(n)$
runs of $\texttt{Refine}$ in each phase. By Lemma \ref{lem:Refine-RunTime},
$\texttt{Refine}$ makes $O(n)$ independence oracle calls. Thus,
each phase takes $O(n^{2})$ independence oracles. Finally, to get
an $(1-\epsilon)$-approximation, we need only $O(1/\epsilon)$ phases,
and thus the total time of this algorithm is a (not too impressive)
$O(n^{2}/\epsilon \cdot\indeptime)$-time. We improve this with a hybrid approach
in the next section. %\begin{thm}
%Our augmenting sets algorithm, which in each phase runs Refine until
%no more element changes its type, runs each phase in $O(n^{2})$ time
%and the entire algorithm in $O(n^{2}/\epsilon)$ time.
%\end{thm}

%\begin{proof}
%We know from Corollary \ref{cor:After-Refine} that in each iteration
%Refine changes the type of at least one element. Since each element
%can change its type at most twice, we know that after $O(n)$ runs
%of Refine we can apply Lemma \ref{subsec:TheAugSetsAlgo} to get  maximal
%augmenting sets. Since by Lemma \ref{lem:Refine-RunTime} Refine takes
%$O(n)$ time, this implies we can execute each phase using $O(n^{2})$
%independence oracles. The theorem follows because we need to execute
%only $O(1/\epsilon)$ phases.
%\end{proof}

\subsection{Going Subquadratic by Combining Augmenting Sets and Augmenting Paths\label{subsec:FasterAlgo}}

We obtain a faster algorithm by exploiting Lemma~\ref{lem:maximal-aug-sets-size-ratio}.
The algorithm is parametrized by an integer $p$ which we set in the
end. The algorithm runs in phases. At the beginning of a phase $\ell$,
the algorithm first invokes Algorithm \ref{alg:getdistances} to get
the layers $D_{1},D_{2},...,D_{2\ell+1}$ in $O(n\log r\cdot\Tind)$
time (Lemma~\ref{lem:Get-Distances-Indep}). Throughout a phase $\ell$
if $S$ is the current solution, then the exchange graph $G(S)$ has
shortest augmenting path length $2(\ell+1)$. After a phase, the shortest
augmenting path length is $\geq2(\ell+3)$. We run for $O(1/\epsilon)$
phases.

In phase $\ell$, for $\ell\leq O(1/\epsilon)$, we run the following
routine.

\begin{algorithm2e}[H]

\caption{$\texttt{Hybrid}(p)$}

\SetAlgoLined

Run $\texttt{Refine}$ until $|B_{1}|-|B_{\ell+1}|\leq p$. Initially,
the LHS can be as large as $n$.

Given the partial augmenting set $\Phi_{\ell}$ at this point, find
an augmenting set $\Pi_{\ell}=(B'_{1},A'_{1},\cdots,B'_{\ell+1})\subseteq\Phi_{\ell}$
satisfying $B'_{\ell+1}=B_{\ell+1}$. This is done using the algorithm
described in Lemma~\ref{lem:partial-to-full-aug-sets}.

Find augmenting paths in $G(S)|\Pi_{\ell}$ as in Section~\ref{sec:Indep-Exact-Algorithm}
till the shortest path length changes.

\end{algorithm2e}

To analyze the running time of the hybrid algorithm, we go through
the following steps. First, we show (Lemma~\ref{lem:l1}) Line 1
takes $O(n^{2}/p)$-independence oracle calls. This is similar to
the argument in the last paragraph of the previous section. Second,
we show (Lemma~\ref{lem:boundingpartial}) that $\Pi_{\ell}$ found
in Line 2 is contained in a {\em maximal} augmenting set $\tPi_{\ell}$
of width $\leq|B_{1}|$. Note that the width of $\Pi_{\ell}$ is $|B_{\ell+1}|$.
Thus, the width of $\tPi_{\ell}\setminus\Pi_{\ell}$ in $G(S)|\Pi_{\ell}$
is $\leq|B_{1}|-|B_{\ell+1}|\leq p$ (because of Line 1). Lemma~\ref{lem:maximal-aug-sets-size-ratio}
then implies {\em any} augmenting set in $G(S)|\Pi_{\ell}$ is
of size $\leq(2\ell+4)p$. Which means, in Line 3 we need to make
at most these many augmentations. Since each augmentation takes $\tilde{O}(n)$
independence oracle calls, Line 3 will run in time $\tilde{O}(np\ell \cdot\indeptime )$.
By selecting $p\approx\sqrt{n}$, we get the desired result.
\begin{lem}
\label{lem:l1} Line 1 of the Hybrid algorithm takes $O(n^{2}/p\cdot\indeptime)$-time. 
\end{lem}

\begin{proof}
By Corollary~\ref{cor:After-Refine}, each call to $\texttt{Refine}$
before Line 1 terminates changes the type of at least $|B_{1}|-|B_{\ell+1}|>p$
elements. Since there are $n$ elements and their type can change
only twice, $\texttt{Refine}$ can only be run for $O(n/p)$ times.
By Lemma~\ref{lem:Refine-RunTime}, each run of $\texttt{Refine}$
makes $O(n)$ independence oracle calls. The lemma follows. 
\end{proof}

\begin{lem}
\label{lem:boundingpartial} Let $\Phi_{\ell}=(B_{1},A_{1},B_{2},\ldots,B_{\ell+1})$
be the partial augmenting sets obtained by running $\texttt{Refine}$
for some number of times. Let $\Pi_{\ell}\subseteq\Phi_{\ell}$ be
an augmenting set. Then, there is a {\em maximal} augmenting set
$\tPi_{\ell}$ containing $\Pi_{\ell}$ and the width of $\tPi_{\ell}$
is at most $|B_{1}|$. %Given augmenting sets $\{A_{k}',B_{k}'\}$ contained in $\{A_{k},B_{k}\}$,
%$\{A_{k}',B_{k}'\}$ is contained in some maximal augmenting sets
%of size at most $|B_{0}|$.
\end{lem}

\begin{proof}
Given $\Pi_{\ell}$, we show a way of applying $\texttt{Refine}$
to get a maximal augmenting set as in Lemma~\ref{lem:maximal-aug-sets},
such that no element in $\Pi_{\ell}$ is removed. Since $B_{1}$ was
a maximal (this is why we run $\texttt{Refine1}(0)$ in the end) set
of elements that can be added while preserving independence, the maximal
augmenting set found can't have width more than $|B_{1}|$. This will
prove the lemma.

We now show how $\texttt{Refine1}$ and $\texttt{Refine2}$ can be
implemented without removing any elements from $\Pi_{\ell}$. Let
$\Pi_{\ell}=(B'_{1},A'_{1},\ldots,A'_{\ell},B'_{\ell+1})$.

$\texttt{Refine1}(k)$ finds a maximal $A\subseteq A_{k}$ s.t. $S-A_{k}+B_{k+1}+A\in\mathcal{I}_{1}$.
Since $S-A_{k}'+B_{k+1}'\in\mathcal{I}_{1}$, $A$ can be chosen so
that $A\cap A_{k}'=\emptyset$ and consequently no element of $A_{k}'$
is removed.

$\texttt{Refine2}(k)$ finds a maximally independent subset of $S-D_{2k}+R_{2k}+B_{k}+F_{2k}=S-A_{k}+B_{k}$
by adding elements of $B_{k}$ first. Since $S-D_{2k}+R_{2k}\subseteq S-A_{k}\subseteq S-A_{k}'$
and $S-A_{k}'+B_{k}'\in\mathcal{I}_{2}$, $B$ can be chosen so that
$B_{k}'\subseteq B$ and consequently no element of $B_{k}'$ is removed.
This finishes the proof of the lemma. %	
%	
%By Lemma \ref{lem:consistaugset}, we can implement the augmenting
%sets algorithm by running Refine repeatedly without removing any element
%of$\{A_{k}',B_{k}'\}$. This maximal augmenting sets would contain
%$\{A_{k}',B_{k}'\}$. Furthermore, we claim that it has at most $|B_{0}|$
%elements.
%Observe that because Refine ends with Refine1(0), $B_{0}$ is a maximal
%subset of non-removed sources that can be added while preserving independence.
%Since removed elements are never selected again, no future $B_{0}$
%can have more element than the current $B_{0}$.
\end{proof}
%Now using the last lemma along with Lemma \ref{lem:partial-to-full-aug-sets},
%we get the following corollary.
%\begin{cor}
%Let $\{A_{k},B_{k}\}$ be the partial augmenting sets obtained by
%running Refine for some number of times. Given augmenting sets $\{A_{k}',B_{k}'\}$
%contained in $\{A_{k},B_{k}\}$, there can be at most $(2l+2)\cdot(|B_{0}|-|B_{0}'|)$
%augmentations in $G(S)|\{A_{k}',B_{k}'\}_{k}$.
%\end{cor}
%Interestingly, Lemma \ref{lem:consistaugset} is needed only ``existentially''.
%It is not necessary to compute $\{A_{k}',B_{k}'\}$ after each invocation
%of Refine to avoid removing their elements.

We now bound the running time of our algorithm. 
\begin{thm}
Phase $\ell$ runs in $O((n^{2}/p+np\ell\log r)\cdot\Tind)$ time
plus the amortized cost $O(\sum_{a\in V}(d_{a}^{\text{end}}-d_{a}^{\text{start}})\log r\cdot\indeptime)$,
where $d_{a}^{\text{end}}$ and $d_{a}^{\text{start}}$ are the distances
of $a$ at the end and start of the phase respectively.
\end{thm}

\begin{proof}
From Lemma~\ref{lem:l1}, we get that Line 1 of the hybrid algorithm
takes $O(n^{2}/p\cdot\Tind)$ time. Let $\tPi_{\ell}$ be the maximal
augmenting set containing $\Pi_{\ell}$ as in Lemma~\ref{lem:boundingpartial}.
Line 2, from Lemma~\ref{lem:partial-to-full-aug-sets} makes $O(n)$
independence oracle calls.

The width of $\tPi_{\ell}\setminus\Pi_{\ell}$ is bounded by $|B_{1}|-|B_{\ell+1}|\leq p$.
Now, consider $G(S)|\Pi_{\ell}$ and let $\calP$ be the largest ordered
collection of consecutive shortest augmenting paths. To find an augmenting
path, we invoke Lemmas \ref{lem:Get-Distances-Indep} and \ref{lem:one-aug-Indep}.
Recall that Lemma \ref{lem:Get-Distances-Indep} updates distances
in $O(\sum_{a\in V}(1+d_{a}-d'_{a})\log r\cdot\indeptime)$ time and,
given distances, Lemma$\ $\ref{lem:one-aug-Indep} finds an augmenting
path in $O(n\cdot\Tind)$ time. Since $d_{a}-d'_{a}$ telescopes,
in Line 3 computing $\calP$ takes time
\[
O(\sum_{a\in V}(d_{a}^{\text{end}}-d_{a}^{\text{start}})\log r\cdot\indeptime)+O(n|\calP|\log r\cdot\Tind).
\]

\global\long\def\wdth{\mathsf{width}}%
 Now, from Theorem~\ref{thm:sap}, we have a maximal augmenting set
$\Pi_{\ell}^{*}$ in $G(S)|\Pi_{\ell}$ with $\wdth(\Pi_{\ell}^{*})=|\calP|$.
From Lemma~\ref{lem:maximal-aug-sets-size-ratio}, we know $\wdth(\Pi_{\ell}^{*})\leq(2\ell+4)\cdot\wdth(\tPi_{\ell}\setminus\Pi_{\ell})$.
Thus, $|\calP|\leq(2\ell+4)\cdot(|B_{1}|-|B_{\ell+1}|)\leq(2\ell+4)p$.
This proves the theorem.
\end{proof}
We now have everything we need to prove our main theorem. 
\begin{proof}[Proof of Theorem~\ref{thm:approxMatrInters}]
 Set $p=\lceil{\sqrt{n\epsilon/\log r}}\rceil$, and run the Hybrid
algorithm until the shortest augmenting path length in $G(S)\geq1/\epsilon$.
This requires at most $O(1/\epsilon)$ phases. Moreover, when applying
Lemma \ref{lem:Get-Distances-Indep} we discard a vertex $a$ as soon
as its distance (lower bound) is greater than $1/\epsilon$. This
is okay as distances are monotonic (Lemma \ref{lem:cunning-monot})
so any such vertex will not belong to any (future) shortest augmenting
paths (or augmenting sets) which have length at most $1/\epsilon$.

First we account for the total amortized cost arising from $O(\sum_{a\in V}(d_{a}^{\text{end}}-d_{a}^{\text{start}})\log r\cdot\indeptime)$.
By the above modification $d_{a}^{\text{end}}\leq1/\epsilon$ so they
sum to at most $O(\frac{1}{\epsilon}n\log r\cdot\indeptime)$ which
is dominated by our desired runtime.

At the beginning of each phase we take $O(n\log r\cdot\Tind)$ time
to construct the layers. Each phase, from the above Theorem, takes
$O(n^{2}/p+np\ell\log r)$ independence oracle calls. Thus the total
time to run phase $\ell$ is $O((n\log r+n^{2}/p+np\ell\log r)\cdot\Tind)=O((n^{2}/p+np\ell\log r)\cdot\Tind)$
time. After all the phases are done, we are left with a set $S\in\I_{1}\cap\I_{2}$
such that the shortest augmenting path length in $G(S)\geq1/\epsilon$.
This implies that $S$ is an $(1-\epsilon)$-approximate solution
by Corollary \ref{cor:approx-algo-certif}. The total running time
is $O\left(\frac{1}{\epsilon}\cdot\left(\frac{n^{2}}{p}+np\log r\frac{1}{\epsilon}\right)\cdot\Tind\right)$.
Since $p=\lceil{\sqrt{n\epsilon/\log r}}\rceil$, we prove the theorem.
\end{proof}

\selectlanguage{american}%
\global\long\def\M{\mathcal{M}}%
\global\long\def\P{\mathcal{P_{\M}}}%
\global\long\def\x{\textbf{\ensuremath{\mathbf{\mathbf{x}}}}}%
\global\long\def\z{\textbf{\ensuremath{\mathbf{z}}}}%
\global\long\def\I{\mathcal{\mathcal{I}}}%
\global\long\def\rk{\mathcal{\textsf{rank}}}%
\global\long\def\circ{\mathcal{\textsf{circuit}}}%
\global\long\def\free{\mathcal{\textsf{free}}}%
\global\long\def\defeq{\stackrel{{\scriptstyle def}}{=}}%
\global\long\def\otime{\mathcal{T}}%
\global\long\def\ranktime{\mathcal{T}_{\rk}}%
\global\long\def\indeptime{\mathcal{T}_{\mathtt{indep}}}%
\global\long\def\exchange{\textsf{exchange}}%
\global\long\def\I{\mathcal{\mathcal{I}}}%
\global\long\def\M{\mathcal{\mathcal{M}}}%
\global\long\def\Rn{\mathcal{\mathbb{R}}^{n}}%
\global\long\def\argmax{\mathrm{argmax}}%
\global\long\def\T{\mathcal{\mathcal{T}}}%
\global\long\def\Pa{\mathcal{P}_{\M_{1}}}%
\global\long\def\P{\mathcal{P}_{\M}}%
\global\long\def\Pb{\mathcal{P}_{\M_{2}}}%
\global\long\def\E{\mathbb{E}}%
\global\long\def\defeq{\overset{\mathrm{def}}{=}}%
\global\long\def\pac{\mathrm{PackNum}}%

\section{Approximately Optimal Fractional Solution via Frank-Wolfe}

\label{sec:FK}

Let $\Pa,\Pb$ be the matroid polytopes for matroids $\M_{1}=(V,\I_{1})$
and $\M_{2}=(V,\I_{2})$ respectively. Let $r$ be the size of their
largest common independent set. It is well-known that the matroid
intersection polytope is precisely $\Pa\cap\Pb$, i.e., $\Pa\cap\Pb$
is the convex hull of (the indicator vectors of) the common independent
sets. The matroid intersection problem can therefore be solved via the
following convex program:
\[
\max_{x\in\Pa\cap\Pb}1^{\top}x\ .
\]

In this section, we apply (a variant of) Frank-Wolfe to solve this
problem faster: Theorem$\ \ref{thm:frac-matr-inters}$ obtain a $(1-\epsilon)$-optimal
\emph{fractional} solution $z$ (i.e., $\sum_{i\in V}z_{i}\geq(1-\epsilon)\cdot r$)
in $O(n^{2}/(r\epsilon^{2})\cdot\indeptime)$ time. To obtain an integral
solution, we show in the next section that upon sampling via $x^{*}$
the size of the ground set can be reduced to $\tilde{O}(r/\epsilon^{2})$.
With this sparser new ground set we can speed up many matroid intersection
algorithms. For instance, combining this with our previous $\tilde{O}(n^{1.5}/\epsilon^{1.5})$
time algorithm would improve the runtime in the regime $r/\epsilon^{2}\leq n\leq r^{2}\epsilon$.

To find an approximate fractional solution efficiently, we relax the
intersection constraint to
\begin{equation}
\max_{x\in\Pa,y\in\Pb}1^{\top}x-\frac{\eta}{2}\|x-y\|^{2}\ ,\label{eq:matroid_intersection_relax}
\end{equation}
where $\eta$ is some parameter to be chosen later. This transformation
is useful as it allows us to solve two disjoint greedy problems in
an iteration of Frank-Wolfe. Our algorithm relies on the following
theorem for constrained convex optimization.
\begin{thm}[{Frank-Wolfe Algorithm$\ $\cite[Theorem 1]{jaggi2013revisiting}}]
\label{thm:frank_wolfe}Given a convex set $K\subset\Rn$ and a convex
function $f$ defined on $K$, suppose that:
\begin{itemize}
\item For any vector $c\in\Rn$, we can find $z_{c}\in\argmax_{z\in K}c^{\top}z$
in time $\T_{K}$.
\item For any vector $z\in\Rn$, we can compute $\nabla f(z)$ in time $\T_{f}$.
\item There is a constant $C_{f}$ such that for any $u,v\in K$ and any
$\gamma\in[0,1]$, we have
\[
f(u+\gamma(v-u))\leq f(u)+\left\langle \nabla f(u),\gamma(v-u)\right\rangle +\frac{\gamma^{2}C_{f}}{2}.
\]
\end{itemize}
In $k$ iterations, we can find a $z\in K$ such that
\[
f(z)\leq\min_{z'\in K}f(z')+\frac{2C_{f}}{k+2}\ .
\]
Each iteration takes $O(n+\T_{K}+\T_{f})$.
\end{thm}

For problem (\ref{eq:matroid_intersection_relax}), we note that to
optimize linear functions of $x$ and $y$ over $\{x\in\Pa,y\in\Pb\}$
we only need the standard greedy method, which can be implemented
using $n$ independence oracle calls. We use this fact to prove the
following main theorem of this section.
\begin{thm}
\label{thm:frac-matr-inters}Given two matroids $\M_{1},\M_{2}$ via
independence oracles, we can find a $z\in\Pa\cap\Pb$ in $O(\frac{n}{r\epsilon^{2}})$
iterations such that $\sum_{i\in V}z_{i}\geq(1-\epsilon)r$, where
$r$ is the size of the largest common independent set and each iteration
takes $O(n\cdot\Tind)$ time.
\end{thm}

\begin{proof}
Ideally, we want to apply Frank-Wolfe directly on the problem
\[
\max_{x\in\Pa,y\in\Pb}1^{\top}x-\frac{\eta}{2}\|x-y\|_{2}^{2}
\]
with large $\eta$. However, the runtime of Frank-Wolfe algorithm
depends implicitly on the diameter of the convex set via the constant
$C_{f}$. To obtain a better $C_{f}$, we truncate both matroids by
the size of largest common independent set $r$. Suppose for now,
we know $\overline{r}$ such that $r\leq\overline{r}\leq2r$. We define
the truncated set by
\[
K=\{(x,y)\in\Pa\times\Pb:\sum_{i}x_{i}\leq\overline{r},\sum_{i}y_{i}\leq\overline{r}\}.
\]
Now, we apply Frank-Wolfe on the function $f(x,y)=-1^{\top}x+\frac{\eta}{2}\|x-y\|_{2}^{2}$
over $K$. Theorem \ref{thm:frank_wolfe} shows that we can find some
$x^{(k)},y^{(k)}$ such that 
\[
f(x^{(k)},y^{(k)})\leq\min_{(x,y)\in K}f(x,y)+\frac{2C_{f}}{k+2}.
\]
We bound $C_{f}$ as follows:
\begin{align*}
C_{f} & =\sup_{u,v\in K,\gamma\in[0,1]}\frac{2}{\gamma^{2}}\left(f(u+\gamma(v-u))-f(u)-\left\langle \nabla f(u),\gamma(v-u)\right\rangle \right)\\
 & =\sup_{u,v\in K,\gamma\in[0,1]}\frac{\eta}{\gamma^{2}}\|\gamma(v-u)\|_{2}^{2}\\
 & \leq\sup_{u,v\in K}\eta\|v-u\|_{1} \qquad \leq \qquad 8\eta r,
\end{align*}
where we used $K\subset[0,1]^{n}$ and $\sum_{i}u_{i}+\sum_{i}v_{i}\leq2\overline{r}\leq4r$.
By considering the largest common independent set as a solution, we
know $\min_{(x,y)\in K}f(x,y)\leq-r$. Hence, from the guarantee of
Frank-Wolfe, we get
\[
-1^{\top}x^{(k)}+\frac{\eta}{2}\|x^{(k)}-y^{(k)}\|_{2}^{2}\leq-r+\frac{16\eta r}{k+2}.
\]

To get a valid fractional solution, we define the vector $z$ by $z_{i}=\min\{x_{i}^{(k)},y_{i}^{(k)}\}$.
Since $x^{(k)}\in\Pa$, $y^{(k)}\in\Pb$, we have that $z\in\Pa\cap\Pb$.
Now, we bound the size of $z$:
\begin{align*}
1^{\top}z & \geq1^{\top}x^{(k)}-\|x^{(k)}-y^{(k)}\|_{1}\\
 & \geq1^{\top}x^{(k)}-\sqrt{n}\cdot\|x^{(k)}-y^{(k)}\|_{2}\\
 & \geq1^{\top}x^{(k)}-\frac{\eta}{2}\|x^{(k)}-y^{(k)}\|_{2}^{2}-\frac{n}{2\eta}\\
 & \geq r-\frac{16\eta r}{k+2}-\frac{n}{2\eta}\ .
\end{align*}
Picking $\eta=\sqrt{\frac{n(k+2)}{32r}}$, we have
\[
1^{\top}z\geq r-\sqrt{\frac{32nr}{k+2}}\ .
\]
Therefore, to get a fractional solution with an additive error of
$\epsilon r$, we can set $k=\frac{32}{\epsilon^{2}}\frac{n}{r}$.

Finally, each iteration involves optimizing a linear function over
$K$ which can be done by the greedy method using $O(n)$ independence
oracle calls. To find $\overline{r}$, we can run the greedy algorithm
to find a maximal common independent set in the two matroids using
$O(n)$ independence oracles. Every maximal solution gives a $1/2$-approximation
to $r$ by Lemma$\ $\ref{lem:augm-path-length} because it implies
each augmenting path is of length at least $3$.\selectlanguage{english}%
\end{proof}

\selectlanguage{american}%

\section{Faster Algorithms by Sparsification\label{sec:sample}}

%\sidford{I changed the title to make it fit on a line, feel free to switch back if you prefer the old.}
From Section \ref{sec:FK}, we are given a point $x\in\Pa\cap\Pb$
in the matroid intersection polytope such that $\|x\|_{1}\geq(1-\epsilon)\cdot r$.
In this section, we show how $x$ can be used to sparsify the matroids
$\M_{1}$ and $\M_{2}$. In particular, we reduce the matroid intersection
problem to solving it on a universe of size $O(\frac{r}{\epsilon^{2}}\log n)$.
We can then use algorithms developed in the previous sections on this
smaller universe to find an approximate integral solution. We use
the following lemma of Karger (recall, $\pac(\M)$ is the maximum
number of disjoint bases in $\M$).
\begin{thm}[{\cite[Theorem 4.1]{karger1998random}}]
\label{KargerLemma}Given a matroid $\M$ with $\pac(\M)=k$, suppose
we sample each element of $\text{\ensuremath{\M}}$ with probability
$p\geq18\ln n\cdot\frac{1}{k\epsilon^{2}}$, yielding a submatroid
$\M(p)$ of $\M$. Then with high probability in $n,$ we have $\pac(\M(p))\in[(1-\epsilon)pk,(1+\eps)pk]$.
\end{thm}

Our strategy to sparsify the matroids is to show that after appropriately
sampling using $x$, the sampled elements $V'$ will have $|V'|\ll n$
and will form a universe $V'$ which is almost ``packed'' for both
the matroids, i.e., $V'$ decomposes into an ``almost perfect'' partition
of independent sets of size $(1-O(\epsilon))r$ for both the matroids.
Taking the average of these independent sets shows that the vector
$\frac{r}{|V'|}\cdot1_{V'}$ is almost independent for both the matroids.
This means our sampling procedure is valid because $||\frac{r}{|V'|}\cdot1_{V'}||_{1}=(1-O(\epsilon))\cdot r$.

We first demonstrate how to sample using $x$ to obtain a new universe
$V'$ with the desired properties in a single matroid $\M$. This
can be seen as a non-uniform version of Karger's result. Note that
any matroid $\M$ naturally induces a new matroid if its elements
are replaced by multiple identical copies (a set with multiple copies
of an element is dependent).
\begin{lem}
\label{lem:sampling}Suppose we are given a point $x\in\P$ with $\|x\|_{1}\geq(1-\epsilon)\cdot r$.
Set $\lambda=\frac{\epsilon r}{4n^{3}}$. Suppose we replace each
element $i\in V$ with $\frac{x_{i}'}{\lambda}:=\lfloor\frac{x_{i}}{\lambda}\rfloor$
identical copies to form a new matroid $\M'$ with universe $V'$.
Let submatroid $\M'(p)$ on universe $V'(p)$ be obtained by independently
sampling each element of $V'$ with probability $p=O\left(\frac{\lambda\ln(n/\epsilon)}{\epsilon^{2}}\right)$.
Then with high probability in $n$,
\begin{enumerate}
\item $V'(p)$ has size at most $(1-O(\epsilon))\cdot rp/\lambda=O\left(\frac{r}{\epsilon^{2}}\ln\frac{n}{\epsilon}\right)$.
\item $V'(p)$ contains $(1-O(\epsilon))\cdot|V'(p)|/r$ disjoint independent
sets of size $(1-O(\epsilon))r$.
\end{enumerate}
\end{lem}

\begin{proof}
Note that $x'\in\P$ is a multiple of $\lambda$ and $\|x'\|_{1}=(1-O(\epsilon))\cdot r$
as the total loss of rounding is at most $n\lambda$.

For (1), the expected size of $V'(p)$ is
\[
\sum_{i}\frac{x_{i}'}{\lambda}\cdot p\ \ =\ \ (1-O(\epsilon))\cdot rp/\lambda\ \ =\ \ O\left(\frac{r}{\epsilon^{2}}\ln\frac{n}{\epsilon}\right),
\]
which holds w.h.p. within a factor of $1\pm\epsilon$ by standard
Chernoff bound.

For (2), we further slightly decrease $x'$ to obtain a vector $\bar{x}\in\P$
so that $\|\bar{x}\|_{1}=\left\lfloor \|x'\|_{1}\right\rfloor :=\bar{r}$.
Clearly $\bar{r}=(1-O(\epsilon))\cdot r$. Now by the Carath�odory's
theorem and the integrality of the matroid (base) polytope, $\bar{x}$
can be written as a convex combination of size-$\bar{r}$ independent
sets,\footnote{This can be seen as follows: truncate $\M$ to size $\bar{r}$. Then
any fractional independent set of size $\bar{r}$ is a convex combination
of independent sets, which are of size at most $\bar{r}$ and hence
exactly $\bar{r}$.}i.e., $\bar{x}=\sum_{i=1}^{n}t_{i}v_{i}$ where $v_{i}$ are indicator
vectors of some size-$\bar{r}$ independent sets and multipliers $t_{i}\geq0$
satisfy $\sum_{i}t_{i}=1$.

We round $t$ down to $t'$ so that it is a multiple of $\lambda$.
Now by considering $t'/\lambda$ copies of $v_{i}$ for all $i$,
we see that $V'$ contains at least
\[
\sum_{i}\frac{t'_{i}}{\lambda}\ \ \geq\ \ \frac{1}{\lambda}-n\ \ =\ \ \left(1-\frac{r\epsilon}{4n^{2}}\right)\frac{1}{\lambda}\ \ =\ \ (1-o(\epsilon))\frac{1}{\lambda}
\]
disjoint size-$\bar{r}$ independent sets. Moreover, the size of set
is $|V'|=\sum_{i}\frac{x_{i}'}{\lambda}=(1-O(\epsilon))\cdot\frac{r}{\lambda}$,
which means $\ln|V'|=O(\ln(n/\epsilon))$. So by Theorem \ref{KargerLemma},
w.h.p. $\M'(p)$ contains at least
\[
(1-\epsilon)p\cdot(1-o(\epsilon))\frac{1}{\lambda}=(1-O(\epsilon))\cdot\frac{|V'(p)|}{r}
\]
disjoint size-$\bar{r}$ independent sets, as desired.
\end{proof}
From (2) of Lemma \ref{lem:sampling}, the universe $V'(p)$ almost
can be partitioned into disjoint independent sets of size $(1-O(\epsilon))\cdot r$.
We now show that $V'(p)$ has a common independent set of size $(1-O(\epsilon))\cdot r$
by putting equal mass on almost every element.
\begin{lem}
\label{lem:sampleindsize}For a given point $x\in\Pa\cap\Pb$ with
$\|x\|_{1}\geq(1-\epsilon)\cdot r$, let $V'(p)$ be sampled as in
Lemma \ref{lem:sampling}. Then $V'(p)$ has a common independent
set of size $(1-O(\epsilon))\cdot r$.
\end{lem}

\begin{proof}
Suppose $V'(p)$ contains disjoint $I_{1}^{(1)},...,I_{k_{1}}^{(1)}\in\mathcal{I}_{1}$
and disjoint $I_{1}^{(2)},...,I_{k_{2}}^{(2)}\in\mathcal{I}_{2}$
all of which are of size $(1-O(\epsilon))\cdot r$. By Lemma \ref{lem:sampling},
we may take $k_{1},k_{2}=(1-O(\epsilon))\cdot|V'(p)|/r$.

We begin with a few simplifications. Let $k=\min\{k_{1},k_{2}\}$
and WLOG assume that all $I_{\alpha}^{(1)}$ and $I_{\beta}^{(2)}$
are of equal size $r'=(1-O(\epsilon))\cdot r$ by dropping elements
from the non-smallest independent sets. For $i=1,2$, define $y^{(i)}\in\mathbb{R}^{V'(p)}$
by
\[
y_{a}^{(i)}=\begin{cases}
1/k & \text{if }a\in I_{1}^{(i)}\cup...\cup I_{k}^{(i)}\\
0 & \text{o.w.}
\end{cases}
\]
Then $y^{(i)}$, which is the average of the indicator vectors for
$I_{1}^{(i)},...,I_{k}^{(i)}$, belongs to the matroid polytope for
$\M'_{i}(p)$. Therefore $y=\min\{y^{(1)},y^{(2)}\}$ belongs to the
matroid intersection polytope for $\M'_{1}(p),\M'_{2}(p)$. Since
the matroid intersection polytope is integral, it suffices to show
that $y$ has size $(1-O(\epsilon))\cdot r$.

We bound the size of $y$ as follows:
\begin{align*}
\|y\|_{1} & \geq\|y^{(1)}\|_{1}+\|y^{(2)}\|_{1}-|V'(p)|/k\\
 & =2r'-r/(1-O(\epsilon))\\
 & =2(1-O(\epsilon))r-(1+O(\epsilon))r \qquad  = \qquad (1-O(\epsilon))r,
\end{align*}
which completes the proof of this lemma.
\end{proof}
Combining the last two lemmas shows that our sampling scheme reduces
the size of the universe to $O\left(\frac{r}{\epsilon^{2}}\ln\frac{n}{\epsilon}\right)$.
\begin{thm}
\label{thm:sparse-to-small}Given a point $x\in\Pa\cap\Pb$ such that
$\|x\|_{1}\geq(1-\epsilon)\cdot r$, we can find in $O(n)$ time a
set $V'$ of size $O(\frac{r}{\epsilon^{2}}\log n)$ such that the
largest common independent subset of $V'$ for $\M_{1}$ and $\M_{2}$
has at least $(1-O(\epsilon))\cdot r$ elements.
\end{thm}

\begin{proof}
From Lemmas \ref{lem:sampling} and \ref{lem:sampleindsize}, $V'(p)$
satisfies all the desired properties. It remains to analyze the runtime
of the sampling scheme. Recall that $V'(p)$ is obtained by sampling
$V'$ which contains identical copies of the elements of $V$. Instead
of creating and sampling $V'$ explicitly, we can equivalently sample
$i\in V$ as a binomial random variable with $\lfloor\frac{x_{i}}{\lambda}\rfloor$
trials and probability $p$.
\end{proof}
\begin{thm}
\label{thm:apx-better-r} \foreignlanguage{english}{There is an $\tilde{O}\left(\left(\frac{n^{2}}{r\epsilon^{2}}+\frac{r^{1.5}}{\epsilon^{4.5}}\right)\cdot\mathcal{T}_{ind}\right)$-time
algorithm to obtain a $(1-\eps)$-approximation to the matroid intersection
problem.}
\end{thm}

\begin{proof}
Using Theorem \ref{thm:frac-matr-inters}, in $O((n^{2}/r\eps^{2})\cdot\Tind)$
time we can find a point $z\in\Pa\cap\Pb$ with $||z||_{1}\geq(1-\eps)r$.
Applying Theorem \ref{thm:sparse-to-small} on $z$, in $O(n)$ time
we can reduce the ground set to $V'$ with $|V'|=O(r\log n/\eps^{2})=:n'$
and the maximum common independent set in $V'$ has size $\geq(1-O(\eps))r$.
Note that any independence oracle query on the matroid restricted
to $V'$ doesn't change, and in particular takes $\Tind$ time. Applying
Theorem \ref{thm:approxMatrInters} on this universe of size $n'$,
we see that in $O((n'^{1.5}/\eps^{1.5})\cdot\Tind)$ time we can get
a common independent set $S$ with $|S|\geq(1-\eps)|S^{*}|$ where
$S^{*}$ is the largest common independent set in $V'.$Thus, $|S|\geq(1-\eps)\cdot(1-O(\eps))r=(1-O(\eps))r$.
That is, $S$ is an $(1-O(\eps))$-approximate solution. The time
complexity follows by setting $n'=O(r\log n/\eps^{2})$ in the above.\selectlanguage{english}%
\end{proof}

\subsubsection*{Acknowledgment}

We are thankful to Kent Quanrud for pointing us to Karger's paper~\cite{karger1998random}\foreignlanguage{american}{.}
We are also thankful to Troy Lee for pointing us to~\cite{Price}, and for his comments on a previous version.

{\small
\bibliographystyle{plain}
\bibliography{bib}
}
\end{document}